%% file: main.tex
\documentclass[a4paper,UKenglish,numberwithinsect, autoref, thm-restate]{oasics-v2021}

\nolinenumbers
\usepackage{hyperref}

\usepackage{array,enumerate}
\usepackage{xcolor,xspace}
\usepackage{amsmath}
\usepackage{dutchcal}
\usepackage{amssymb} 
\usepackage{tikz}
\usetikzlibrary{automata, arrows.meta, positioning}
\usepackage{yhmath} 
\usepackage{minted} 

\usepackage[textsize=scriptsize,textwidth=4.3cm]{todonotes}
\usepackage{stmaryrd}
\usepackage{setspace}

\usepackage[capitalize,nameinlink]{cleveref}

\usepackage{caption}
\usepackage{subcaption}
\input{macros}

\usepackage{stmaryrd}

\newcommand{\ie}{\mbox{i.e.,\@}\xspace}%
\newcommand{\eeg}{\mbox{e.g.,\@}\xspace}%
\newcommand{\resp}{resp.\@\xspace}%
\newcommand{\nonl}{\renewcommand{\nl}{\let\nl\oldnl}}

\newcommand{\logicsym}[1][]{\ensuremath{\mathit{{#1}T^{lin}}}\xspace}
\newcommand{\logic}{\logicsym[]}%
\makeatletter
\newcommand\bigDiamond{\mathop{\mathpalette\bigDi@mond\relax}}
\newcommand\bigDi@mond[1]{%
  \vcenter{\hbox{\m@th
    \scalebox{\ifx#1\displaystyle 2\else1.2\fi}{$#1\Diamond$}%
  }}%
}
\newcommand\bigLozenge{\mathop{\mathpalette\bigL@zenge\relax}}
\newcommand\bigL@zenge[2]{%
  \vcenter{\hbox{\m@th
    \scalebox{\ifx#1\displaystyle 2\else1.2\fi}{$#1\blacklozenge$}%
  }}%
}
\makeatother

\title{\centerline{Revisiting the expressiveness of metric temporal logic:} \centerline{A tale of ``Je t'aime, moi non plus.''}}

  \author{Mohammed Aristide Foughali}{IRIF, Universit\'{e} Paris Cit\'{e}, France}{}{}{} 
\authorrunning{M.A. Foughali}
\titlerunning{Revisiting the expressiveness of metric temporal logic}

\begin{document}
\maketitle

\input{0.abstract}
\input{1.intro}
\input{2.preliminaries-linear-time}

\input{3.logics}

\input{4.monitors}

\input{8.conclusion}

\bibliographystyle{plain}
\bibliography{biblio}
\newpage
\appendix
\input{incomparability}
\input{mixed}


\end{document}

%% file: macros.tex
\newcommand{\BNFsep}{\;\;|\;\;}

\newcommand{\sem}[1]{\llbracket #1 \rrbracket}

\newcommand{\act}{\Sigma}

\newcommand{\rplus}{\mathbb{R}_{\geq 0}}

\newcommand{\icompile}[1]{\llparenthesis #1 \rrparenthesis_{\itww}}
\newcommand{\pcompile}[1]{\llparenthesis #1 \rrparenthesis_{\pww}}

\newcommand{\trho}{\ensuremath{\tilde{\rho}}}
\newcommand{\tsig}{\ensuremath{\tilde{\sigma}}}
\newcommand{\ttau}{\ensuremath{\tilde{\tau}}}

\newcommand{\pww}{\ensuremath{\mathit{pw}}}

\newcommand{\itww}{\ensuremath{\mathit{itw}}}
\newcommand{\itss}{\ensuremath{\mathit{its}}}

\newcommand{\mxx}{\ensuremath{\mathit{mx}}}

\newcommand{\pw}{\ensuremath{\models_{\pww}}}

\newcommand{\its}{\ensuremath{\models_{\itss}}}
\newcommand{\itw}{\ensuremath{\models_{\itww}}}

\newcommand{\npw}{\ensuremath{\nvDash_{\pww}}}

\newcommand{\nitw}{\ensuremath{\nvDash_{\itww}}}

\newcommand{\mx}{\ensuremath{\models_{\mxx}}}
\newcommand{\nmx}{\ensuremath{\nvDash_{\mxx}}}
\newcommand{\mtl}{\ensuremath{\mathit{MTL}}}

\newcommand{\lf}[2]{\ensuremath{\sem{\mathit{#1}}_{\mathit{#2}}^{f}}}
\newcommand{\ilf}[2]{\ensuremath{\sem{\mathit{#1}}_{\mathit{#2}}^{\mathit{inf}}}}
\newcommand{\elf}[2]{\ensuremath{\sem{\mathit{#1}}_{\mathit{#2}}^{\mathit{fi}}}}
\newcommand{\alf}[2]{\ensuremath{\mathit{\sem{#1}_{#2}}}}

\newcommand{\Kappa}{\mathrm{K}}
\newcommand{\Rho}{\mathrm{P}}
\newcommand{\tRho}{\mathit{\tilde{\Rho}}}
\newcommand{\nxt}[2]{\bigcirc_{#1}\, #2}

\newcommand{\until}[3]{\ensuremath{{#1}\, \, \mathcal{U}_{#2}\, {#3}}}

\newcommand{\uuntil}[2]{\ensuremath{{#1}\, \, \mathcal{U}\, {#2}}}
\newcommand{\globally}[2]{\ensuremath{{\square}_{#1}\, {#2}}}
\newcommand{\uglobally}[1]{\ensuremath{{\square}\, {#1}}}
\newcommand{\eventually}[2]{\ensuremath{{\textstyle\bigDiamond}_{#1}\, {#2}}}
\newcommand{\ueventually}[1]{\ensuremath{{\textstyle\bigDiamond}\, {#1}}}

\newcommand{\naturals}{\mathbb{N}}

\newcommand{\bvSym}[1]{\mathit{bv}}


%% file: 0.abstract.tex
\begin{abstract}

The expressiveness of Metric Temporal Logic (MTL) has been extensively studied throughout the last two decades. %
In particular, it has been shown that the \emph{interval-based} semantics of MTL is strictly more expressive than the \emph{pointwise} one. %
These results may suggest that enabling the evaluation of formulae at arbitrary time points \emph{instead of} positions of timed events  increases the expressive power of MTL. %
In this paper, we formally argue otherwise. %
We demonstrate that under standard models of finite or non-Zeno infinite (action-based) timed executions, the interval-based and the pointwise semantics are incomparable. %
We then propose a new \emph{mixed} semantics 
that embeds both the pointwise and the interval-based ones. %

\ccsdesc{Theory of computation~Logic and verification}
\keywords{Timed logics, Metric temporal logic}

\end{abstract}

%% file: 1.intro.tex
\section{Introduction}
\label{sec:intro}

Metric Temporal Logic (MTL)~\cite{koymans1990specifying} is a popular \emph{timed} extension of the Linear-time Temporal Logic LTL~\cite{pnueli1977temporal}. %
The importance of MTL motivated a large body of research investigating its expressive power~\cite{d2007expressiveness,bouyer2005expressiveness,bouyer2010expressiveness,haase2010process,ouaknine2008some,furia2007expressiveness}. %
In this paper, we focus on the expressiveness of the \emph{pointwise} semantics of MTL as compared to its \emph{interval-based} one. %
In order to present the existing results on such comparison, we first (informally) introduce the model of \emph{timed words}. %

A \emph{timed event} is a couple containing a letter, representing an \emph{action}, and a real number corresponding to its \emph{timestamp}, \ie the global time at which it occurred. A timed word is a finite or infinite sequence of timed events. %
Following the least constrained model of timed words~\cite{alur1994determinizable,ouaknine2005decidability,ouaknine2007decidability}, which we refer to hereafter as the \emph{general model}, a (finite or infinite) timed word satisfies time monotonicity, \ie timestamps are non-decreasing, and an infinite timed word obeys time divergence (also called non-Zenoness), that is global time does not converge towards some bounded value. %
For instance, over the alphabet $\{a,b,c\}$, the sequence $(a,0)(b,1)(c,1)(a,5.3)$ is a finite timed word whereas $(a,0)(b,1)(c,0.8)$ is not (time monotonicity is violated). %
Similarly, the infinite timed word $(a,\frac{1}{2})(b,\frac{3}{4})(b,\frac{7}{8})(c,\frac{15}{16})\ldots$, \ie where the difference between two consecutive timestamps is halved at every next action, is excluded (global time converges towards $1$). %

Under the pointwise semantics, formulae are evaluated at the \emph{positions} of timed events within the timed word, whereas the interval-based semantics is given over continuous intervals within, typically, a timed \emph{signal}~\cite{ouaknine2008some}. %
To enable the comparison between both semantics, %
D'Souza and Prabhakar~\cite{d2007expressiveness} propose to interpret the interval-based semantics over timed words (instead of timed signals), allowing the evaluation of formulae at arbitrary time points (including those where no events occur, \ie not corresponding to a timestamp). %
They demonstrate accordingly that the interval-based semantics is strictly more expressive than the pointwise one over finite timed words, using a constrained model where time monotonicity is strict (\eeg the timed word $(a,0)(b,1)(c,1)(a,5.3)$ used as an example above is excluded by design). %
We refer to this existing result as \textbf{ER}:

\begin{center}
\textbf{ER}:  The interval-based semantics is strictly more expressive\\ than
 the pointwise one over strictly monotone finite timed words.
\end{center}
 
However, we argue that while strict monotonicity of time is reasonable is many cases, it is not justified in general. %
Indeed, simultaneous interleaving in timed systems, theorised  since~\cite{dill1989timing,alur1991logics}, is of sheer importance when formally reasoning about concurrent behaviour, \eeg in multiprocessor real-time scheduling (see for instance~\cite{DBLP:journals/jsa/FoughaliHZ23,brandenburg2011scheduling,DBLP:conf/compsac/FoughaliMZ24}), and distributed applications~\cite{liebig1999event,fidge2002logical,lamport2019time}. %

In a remotely related work, Bouyer et al. proved in two seminal papers~\cite{bouyer2005expressiveness,bouyer2010expressiveness} that TPTL~\cite{alur1994really} is strictly more expressive than MTL. %
These two publications~\cite{bouyer2005expressiveness,bouyer2010expressiveness} are of particular interest to us for the following two reasons: 
\begin{itemize}
\item{} They use a clever encoding of timed words into \emph{timed state sequences} to reason over the interval-based semantics. This encoding will inspire our incomparability proofs (contributions \textbf{C1}, \textbf{C2} and \textbf{C3} below). 
\item{} While considering the general model of infinite timed words~\cite[Sect. 2]{bouyer2005expressiveness}\cite[Sect. 2]{bouyer2010expressiveness}, they make the following side claim~\cite[Section 1]{bouyer2010expressiveness}\cite[Sect. 1]{bouyer2005expressiveness}:

\begin{center}
``As side results, we get that MTL is strictly more expressive under \\ the interval-based semantics than under the pointwise one''
\end{center}

We formally show that this claim does not hold (\cref{thm:inc}). %
\end{itemize}

There is therefore a need to complete the picture of \textbf{ER} and, in case of incomparability of the interval-based and the pointwise semantics, unify them in order to benefit from their respective advantages. %

\textbf{Contributions.} In this paper, we (i) answer the comparability question by the negative under the general model of timed words and (ii) devise a new mixed semantics that remedies such incomparability. %
More precisely, we provide the following contributions: 
\begin{itemize}
\item{\textbf{C1}} The interval-based and the pointwise semantics are incomparable under the general model of finite timed words, %
\item{\textbf{C2}} The interval-based and the pointwise semantics are incomparable under the general model of infinite timed words,
\item{\textbf{C3}} The incomparability holds even under reasonable restrictions such as \emph{non-stuttering},
\item{\textbf{C4}} A mixed semantics embedding both the interval-based and the pointwise ones is proposed. 
\end{itemize} 
\textbf{C1}, \textbf{C2} and \textbf{C3} complete the picture of \textbf{ER}, showing that it does not hold under the more general (non-strict) monotonicity, regardless of whether timed words are finite or infinite. %
In addition, the mixed semantics (\textbf{C4}) conjugates the expressive powers of the interval-based and the pointwise ones, while embedding both semantics through automatic compilers.

\textbf{Outline.} %
The rest of this paper is organised as follows. %
\cref{sec:preliminaries} is dedicated to preliminaries, whereas the core of our contributions is detailed in \cref{sec:logics} (contributions  \textbf{C1}, \textbf{C2} and \textbf{C3}) 
and~\cref{sec:monitors} (contribution \textbf{C4}). %
~\cref{sec:conclusion} wraps up with concluding remarks and future work directions. %

%% file: 2.preliminaries-linear-time.tex
\section{Preliminaries}
\label{sec:preliminaries}

\subsection{Binary relations and functions}

Let $X$ and $Y$ be two non-empty sets. 
\begin{definition}[Binary relation]
A binary relation $\mathcal{R}$ over $X$ and $Y$ is a subset of the cartesian product 
$X \times Y$. %
The inverse relation $\mathcal{R}^{-1}$ of $\mathcal{R}$ is defined as $\mathcal{R}^{-1} = \{(y,x)\in Y \times X \, \, | \, \, (x,y)\in \mathcal{R}\}.$

For any $x\in X$, we let $\mathcal{R}[x] = \{y\in Y\, \, |\, \, (x,y)\in\mathcal{R}\}$ denote the set of all \emph{images} of $x$ by $\mathcal{R}$. %
$\mathcal{R}$ is said left total iff every element of $X$ has at least an image by $\mathcal{R}$, 
\ie $\forall x\in X: \mathcal{R}[x]\neq \emptyset$, 
right total iff $\mathcal{R}^{-1}$ is left total, and total if it is both right and left total. %
\end{definition}

\begin{definition}[Function]
A binary relation $\mathcal{R}$ over $X$ and $Y$ is called a \emph{function} from $X$ to $Y$ iff every element in $X$ has exactly one image in $Y$, \ie 
$\forall x\in X: |\mathcal{R}[x]| = 1$. 
We use the standard notation for functions, \ie $\mathcal{R}: X\rightarrow Y$ for ``$\mathcal{R}$ a function from $X$ to $Y$'' and  $\mathcal{R}(x) = y$ (instead of $\mathcal{R}[x] = \{y\}$) for ``$y$ is the image of $x$ by $\mathcal{R}$''. %
\end{definition}

\begin{definition}[Injective, surjective and bijective functions]\label{def:inj}
A function $\mathcal{R}$ is said:
\begin{itemize}
\item{} \emph{injective} iff all images are distinct: $\forall y\in Y \, \,  :
|\mathcal{R}^{-1}[y]| \leq 1$,
\item{} \emph{surjective} iff it is right total ($\mathcal{R}^{-1}$ is therefore a total relation), %
\item{} \emph{bijective} iff it is both injective and surjective ($\mathcal{R}^{-1}$ is therefore a bijective function). 
\end{itemize}
\end{definition}

\subsection{Timed words and timed state sequences}
Let $s$ be a sequence over a non-empty set $X$. %
We let sequences start at position $0$, and $|s|$ denote the length of $s$ defined in the usual way, with $|s| = \omega$ if $s$ is infinite. %

\subsubsection{Timed words}
The definitions provided in this subsection follow the general model of timed words, they can be seen as a generalisation of those given in~\cite{bouyer2005expressiveness} to both finite and infinite timed words.

\begin{definition}[Time sequence]\label{def:ts}
A \emph{time sequence} $\tau$ is a non-empty sequence over $\rplus$ satisfying the following conditions:
\begin{itemize}
\item{Monotonicity:} $\forall  i \in \naturals, 0 < i < |\tau| : \tau_{i-1} \leq \tau_{i}$,
\item{Time divergence:} $|\tau| = \omega \Rightarrow (\forall t \in \rplus\, \,  \exists i\in \naturals: \tau_i > t)$.
\end{itemize}
\end{definition}

\begin{remark} In \cref{def:ts}, time divergence is only required for infinite time sequences. This is because every finite time sequence satisfying monotonicity is a prefix of a time-divergent infinite time sequence~\cite{alur1997modularity}. %
\end{remark}

Let $\act$ be a non-empty, finite set of letters, $\act^+$ and $\act^\omega$ denote, respectively, the set of non-empty finite and infinite words over $\Sigma$, and $\act^\infty = \act^+ \cup \act^\omega$. %

\begin{definition}[Timed word]\label{def:tw}
A timed word is a pair $\rho = (\sigma,\tau)$ with $\sigma\in\act^\infty$ and $\tau$ a time sequence such that $|\sigma| = |\tau|$. %
The length of $\rho$ is $|\rho| = |\sigma|$, and each $(\sigma_i,\tau_i)$ is called a timed event. %
For some $i\in \naturals$, $i < |\rho|$, $\tau_i$ is called the \emph{timestamp} of $\sigma_i$. %
\end{definition}

Accordingly, if $\rho = (\sigma,\tau)$ is finite, say of length $n\in\naturals_{>0}$ (\resp infinite), $\sigma = \sigma_0\sigma_1\ldots\sigma_{n-1}$ is a finite word in $\act^+$ and $\tau = \tau_0\tau_1\ldots\tau_{n-1}$ is a finite time sequence (\resp $\sigma = \sigma_0\sigma_1\ldots$ is an infinite word in $\act^\omega$ and $\tau = \tau_0\tau_1\ldots$ is an infinite time sequence). %
The duration of $\rho$, denoted $\mu(\rho)$, equals the timestamp of the last letter if $\rho$ is finite, \ie $\mu(\rho) = \tau_{|\rho| - 1}$. %
Otherwise, $\mu(\rho) = \infty$ thanks to time divergence (\cref{def:ts}). %
We sometimes use the alternative notation of sequences of timed events (\eeg $(\sigma_0,\tau_0)(\sigma_1,\tau_1)\ldots(\sigma_{|\rho|-1},\tau_{|\rho|-1})$ instead of $(\sigma,\tau)$ for some finite timed word). %
We denote by $^m\rho$, $m<|\rho|$, the prefix of $\rho = (\sigma,\tau)$ up to and including position $m$, \ie $^m\rho = (\sigma_0,\tau_0)\ldots(\sigma_m,\tau_m)$.

\begin{remark}
A timed word $\rho = (\sigma,\tau)$ inherits the properties of the time sequence $\tau$, \ie time monotonicity and, if $\rho$ is infinite, time divergence (\cref{def:ts}). %
\end{remark}

\begin{example}\label{ex1} Let $\Sigma = \{a,b,c\}$. The sequence $\rho_1 = (a,1)(b,1.3)(c,3.5)(b,3.5)$ is a timed word of length $4$ and duration $3.5$. %
The sequence $\rho_2 = (a,1)(b,1.3)(c,1.2)(b,3.5)$ is not a timed word (violation of time monotonicity). %
The sequence $\rho_3 = (\sigma_i,\tau_i)$, $i\in \naturals$ defined next is an infinite timed word:  $\sigma_{2i} = a$, $\sigma_{2i+1} = b$; $\tau_0 = 0$ and for all $i$ in $\naturals_{>0}$: $\tau_{i} = \tau_{i-1} + 1$ if $i\bmod 3 = 0$ and $\tau_{i} = \tau_{i - 1}$ otherwise. %
Its prefix $^6\rho_3$ equals $(a,0)(b,0)(a,0)(b,1)(a,1)(b,1)(a,2)$. %
 \end{example}

We now define some subclasses of timed words. %
A strongly stutter-free (\resp strictly monotone) timed word is exempt of the behaviour where two occurrences of the same letter (\resp of any two letters) happen at different positions with the same timestamp. %

\begin{definition}[Stutter freeness and strict monotonicity]\label{def:stut}

A timed word $\rho = (\sigma,\tau)$ is said: 
\begin{enumerate}
 \item{} \emph{stutter free} iff $\forall i\neq j\in \naturals, i,j<|\rho|: \tau_i = \tau_{j} \Rightarrow \sigma_i \neq  \sigma_{j}$.\label{def:stut1} 
 \item{} \emph{strictly monotone} iff $\tau$ is strictly monotone: $\forall  i \in \naturals, 0< i < |\tau| : \tau_{i-1} < \tau_{i}$.\label{def:stut2}  
 \end{enumerate}
 \end{definition}

\begin{remark} While strict monotonicity implies stutter freeness, none of them is related to time-divergence. %
For instance, the infinite timed word $\rho_3$ of~\cref{ex1} is time divergent but not stutter free (and therefore not strictly monotone). %
\end{remark}

\begin{example} $(a,0)(c,3)$ is strictly monotone (and therefore stutter free), $(a,0)(b,1)(a,1)(c,3)$ is stutter free, and $(a,0)(b,1)(a,1)(b,1)(c,3)$ is not stutter free. %
\end{example}

 \subsubsection{Timed state sequences}\label{sec:tss}
 
The definitions given in this subsection generalise the model of infinite timed state sequences given in~\cite{bouyer2005expressiveness,bouyer2010expressiveness} to both the finite and infinite cases.  
 
A \emph{non-empty interval} $I$ over $\rplus$ has the form $[l,u]$, $(l,u]$, $[l,u)$ or $(l,u)$ with its left bound $I^l = l\in \rplus$,  right bound $I^u = u\in \rplus\cup\{\infty\}$, $I^l\leq I^u$ if $I$ is right- and left-closed (\ie has the form $[l,u]$) and $I^l < I^u$ otherwise,  and $\infty]$ excluded; let $\mathcal{I}$ be the set of all such intervals. %
An element of $\mathcal{I}$ denotes a non-empty convex subset of $\rplus$ defined uniquely in the usual way. Two intervals $I_1,I_2\in \mathcal{I}$ are \emph{adjacent} iff $I_1\cap I_2 = \emptyset$, $I_1\cup I_2 \in \mathcal{I}$ and $x<y$ for every $x\in I_1$ and $y\in I_2$. %
An interval $I$ is \emph{punctual} iff it denotes a singleton set, \ie it has the form $[l,l]$. %

\begin{definition}[Interval sequence] \label{def:is}
An interval sequence $\mathcal{\iota}$ is a non-empty sequence over $\mathcal{I}$ satisfying:
\begin{itemize}
\item{Adjacency:} $\forall i \in \naturals, 0 < i < |\mathcal{\iota}| : \iota_{i-1} \text{ and } \iota_{i} \text{ are adjacent}$,
\item{Progress:} $|\mathcal{\iota}| = \omega \Rightarrow \forall t\in \rplus\, \, \exists i\in \naturals: t\in \iota_i$.
\end{itemize} 
\end{definition}

\begin{definition}[Timed state sequence] Let $\mathit{AP}$ be a finite set of atomic propositions. 
A timed state sequence over $\mathit{AP}$ is a pair $\kappa = (\eta,\iota)$ where $\eta$ is a sequence over $\mathcal{P}(AP)$ (the powerset of $\mathit{AP}$), $\iota$ an interval sequence and $|\eta| = |\iota|$. %
\end{definition}

Similarly to timed words, given a timed state sequence $\kappa = (\eta,\iota)$, we (i) sometimes use the alternative sequence notation $(\eta_i,\iota_i)$, $i\in\naturals$, $i<|\kappa|$ instead of $(\eta,\iota)$ and  %
 (ii) define $\mu(\kappa)$, the duration of $\kappa$ as $\iota^u_{|\kappa| -1}$ if $\kappa$ is finite (requiring, without loss of generality, $\iota_{|\kappa| -1}$ to be right closed), and $\infty$ otherwise (thanks to the progress condition, \cref{def:is}).  %
 Moreover, given $t$ in $\rplus$ (satisfying $t\leq \mu(\kappa)$ if $\kappa$ is finite), we denote by $\kappa(t)$ the set $\eta_i$ where $t\in \iota_i$ (such $\iota_i$ is uniquely defined thanks to the adjacency property). 

\begin{example}\label{ex:ts}
Let $\mathit{AP} = \{p,q\}$.  The sequence $(\{p\},[1,2))(\{p,q\},[2,3])(\{p\},[3,3.4])$ is not a timed state sequence (violation of interval adjacency). The sequence \\%
 $\kappa = (\{p\},[1,2))(\{p,q\},[2,3))(\{q\},[3,3]),(\{p\},(3,3.4])$ is a timed state sequence with length $4$ and duration $3.4$, and we have $\kappa(1.2) = \kappa(3.1) = \{p\}$, $\kappa(2.2) = \{p,q\}$ and $\kappa(3) = \{q\}$.
\end{example}

\subsection{Timed words vs. timed state sequences}\label{sec:vs} A timed word is a natural representation of an \emph{action-based} execution of a timed system. %
For example, $(a,2)(b,2)(c,4.1)$ corresponds to an execution where the system waited for two time units, performed action $a$ then immediately $b$, then waited for 2.1 time units and finally performed $c$. %
A timed state sequence, on the other hand, is suitable for signals: %
\eeg $(\{p\},[0,2))(\{q,r\},[2,3])$ represents a signal where $p$ held for two time units (excluding $2$), then $q$ and $r$ held for one time unit. %
While there is a two-way translation between the untimed versions of these models modulo bisimulation~\cite{de1990action}, this is not the case in the timed setting, see \eeg the discussion on signals vs. timed words in~\cite{asarin2002timed} (this heterogeneity led to the proposition of signal-event frameworks as in~\cite{berard2006refinements}). %
We adopt an action-based view for the simple reason that it is the only way to enable any comparison involving the pointwise semantics. %
We still need timed state sequences as their use to encode timed words
~\cite{bouyer2005expressiveness,bouyer2010expressiveness} inspire our incomparability proofs. %

\subsection{Metric Temporal Logic (MTL)}\label{sec:mtl}
Let $\mathcal{I}_q$ be the set of all non-empty intervals over $\rplus$ with rational endpoints, \ie defined similarly to $\mathcal{I}$ (\cref{sec:tss}) but with $l\in\mathbb{Q}_{\geq 0}$, $u\in \mathbb{Q}_{\geq 0}\cup\{\infty\}$. %
In order to define clean syntax and semantics of MTL under both the pointwise and the interval-based semantics, we equate hereafter the sets of atomic propositions $\mathit{AP}$ and letters (actions) $\act$ and refer to both as $\act$. %
The syntax of MTL is given below with $a\in \act$ and $I\in \mathcal{I}_q$. %

  $$
    \begin{array}{lll}
      \phi,\psi &:=& a \BNFsep \neg \phi \BNFsep \phi \wedge \psi \BNFsep  
        \until{\phi}{I}{\psi}  
    \end{array}
  $$

Where $\until{}{I}{}$ denotes the \emph{timed Until} operator (note that we use its standard strict version~\cite{bouyer2005expressiveness,ouaknine2008some,ouaknine2005decidability}). %
The syntactic sugar is defined as usual: $\phi \vee \psi \equiv \neg(\neg\phi \wedge \neg\psi)$, $\top \equiv a \vee \neg a$, $\nxt{I}{\phi} \equiv \until{(\neg\bigvee_{a\in\Sigma}a)}{I}{\phi}$ (\emph{next} $\phi$), $\eventually{I}\phi \equiv \until{\top}{I}{\phi}$ (\emph{eventually} $\phi$) and $\globally{I}{\phi} \equiv \neg\eventually{I}{\neg\phi}$ (\emph{globally} $\phi$). %
We will use often the handy formula $\varphi_\Sigma = \bigvee_{a\in\Sigma}a$ (utilised to define the next operator above) hereafter. %
We omit the subscript $I$ when $I = [0,\infty)$. %

Let $\rho = (\sigma,\tau)$ be a timed word and $\kappa = (\eta,\iota)$ a timed state sequence. %
Following~\cite{bouyer2005expressiveness,d2007expressiveness}, we force $\tau_0$ to equal $0$ and $\iota_0$ to be left closed at $0$ without loss of generality. %
We let $\Rho$ and $\Kappa$ be, respectively, the set of all such timed words and timed state sequences. %
We give below the pointwise semantics (over elements of $\Rho$) and the two versions of the interval-based semantics (one over $\Rho$ as in~\cite{d2007expressiveness}, and one over $\Kappa$ following~\cite{bouyer2005expressiveness}).  
  
  \subsubsection{Pointwise semantics} Formulae are evaluated over a timed word $\rho = (\sigma,\tau)\in \Rho$ and a position $i \in \naturals$ with $i < |\rho|$. %
  To avoid unnecessary cluttering of notation, $j < |\rho|$ is assumed in the semantics of $\until{\phi}{I}{\psi}$. %
  $$
  \begin{array}{lll}
      \rho,i\pw a &\text{ iff }& \sigma_i = a
      \\
      \rho,i\pw \neg \phi &\text{ iff }& \rho,i\npw \phi
      \\
      \rho,i\pw \phi \wedge \psi &\text{ iff }& \rho,i\pw \phi \text{ and } \rho,i\pw \psi
      \\
      \rho,i\pw \until{\phi}{I}{\psi} &\text{ iff} & \exists j > i: \tau_j-\tau_i\in I \text{ and } \rho,j\pw \psi \text{ and } \forall i < k < j : \rho,k\pw \phi%
    \end{array}
$$

\subsubsection{Interval-based semantics}

\textbf{Over timed words.} Formulae are evaluated over a timed word $\rho = (\sigma,\tau)\in \Rho$ and a timepoint $t \in \rplus$ (with $t\leq \mu(\rho)$ if $\rho$ is finite). %
To avoid heavy notation, $|\rho| \neq \omega \Rightarrow t'\leq \mu(\rho)$ is assumed in the semantics of $\until{\phi}{I}{\psi}$. %
   
   $$
  \begin{array}{lll}
      \rho,t\itw a &\text{ iff }& \exists i < |\rho| : \sigma_i = a \text{ and } \tau_i = t
      \\
      \rho,t\itw \neg \phi &\text{ iff }& \rho,t\nitw \phi
      \\
      \rho,t\itw \phi \wedge \psi &\text{ iff }& \rho,t\itw \phi \text{ and } \rho,t\itw \psi
      \\
      \rho,t\itw \until{\phi}{I}{\psi} &\text{ iff} & \exists t' > t  : t'-t\in I \text{ and } \rho,t'\itw \psi \text{ and } \forall t < t'' < t' : \rho,t''\itw \phi%
    \end{array}
$$

\textbf{Over timed state sequences.} Formulae are interpreted over a timed state sequence $\kappa = (\eta,\iota)\in \Kappa$ and a timepoint $t \in \rplus$ (with $t\leq \mu(\kappa)$ if $\kappa$ is finite). We give the definition for $\phi = a$ only, the remaining cases can be obtained from $\models_{itw}$ by replacing $\rho$ with $\kappa$. %

$$
  \begin{array}{lll}
      \kappa,t\its a &\text{ iff }& a \in \kappa(t)
    \end{array}
$$

We say that $\rho$ satisfies $\phi$, written $\rho \models_{sem} \phi$ iff $\rho,0 \models_{sem} \phi$ with $\mathit{sem\in \{pw,itw\}}$ and similarly $\kappa \models_{its} \phi$ iff $\kappa,0 \models_{its} \phi$. %
Let $\mathit{sem\in \{pw,itw,its\}}$. %
We write $\lf{\phi}{sem}$ (\resp $\ilf{\phi}{sem}$) to denote the set of all finite (\resp infinite) sequences $s$ such that $s\models_{\mathit{sem}} \phi$, and let $\alf{\phi}{sem} = \lf{\phi}{sem} \cup \ilf{\phi}{sem}$. %
For instance, $\alf{\phi}{pw} = \{\rho\in \Rho\, \, | \, \, \rho \pw \phi\}$. %
When $\mathit{sem\in \{pw,itw\}}$, each of $\lf{\phi}{sem}$, $\ilf{\phi}{sem}$ and $\alf{\phi}{sem}$ defines a \emph{timed language}. %

\begin{example}\label{exincomp} Let $\rho_1 = (a,0)(b,1)(a,1)(c,3.3)$, $\zeta_1 = \ueventually{(b \wedge \nxt{[0,0]}{a})}$ and $\zeta_2 = \eventually{(0,1)}{\eventually{[0,3.5]}{c}}$. %
For the formula $\zeta_1$, that reads ``eventually $b$ will occur followed immediately by an $a$'', we have $\rho_1 \pw \zeta_1$ but $\rho_1 \nitw \zeta_1$ (the closest formula to $\zeta_1$ under \itw\,  is $\ueventually{(a\wedge b)}$ where the order is lost, see below). %
On the other hand, for $\zeta_2$ that reads ``from some time within $(0,1)$, $c$ will occur between $0$ and $3.5$ time units'' we have  $\rho_1 \itw \zeta_2$ but $\rho_1 \npw \zeta_2$. %
\end{example}
The above example is an excellent motivator for the remainder of this paper: its counterintuitive outcome (as $\rho_1$ ``should'' satisfy both formulae) owes to the fact that the pointwise semantics can reason only at positions of timed events, whereas the interval-based one loses the notion of order over simultaneous events as we will see throughout~\cref{sec:logics}.

%% file: 3.logics.tex
\section{The Pointwise and Interval-Based Semantics are Incomparable}
\label{sec:logics}

\subsection{Main observation}\label{sec:int}

\textbf{Over finite words.} To explain the intuition behind the incomparability of both semantics, an excellent starting point is the existing result \textbf{ER} (\cref{sec:intro}), since D'Souza and Prabhakar~\cite{d2007expressiveness} use a clean version of pointwise and interval-based semantics both interpreted over finite timed words, \ie the satisfaction relations $\models_{pw}$ and  $\models_{itw}$ (\cref{sec:mtl}) where $|\rho| < \omega$. %
To arrive at \textbf{ER}, the authors use the formula $\phi_{\mathit{dp}}$ defined below\footnote{We slightly rewrote this formula in terms of the standard (strict) Until used in this paper.} over $\Sigma = \{a,b\}$  and prove formally that there exists no formula $\psi$ such that $\lf{\psi}{pw} = \lf{\phi_{\mathit{dp}}}{itw}$. 

$$\phi_{\mathit{dp}} = \neg\ueventually{(a\wedge(\uuntil{\neg \varphi_\Sigma}{a}) \wedge (\uuntil{\neg a}{(\neg a \wedge \eventually{[1,1]}{(a \vee b)})}))}$$

The language $\lf{\phi_{\mathit{dp}}}{itw}$ contains all finite timed words such that for every two consecutive occurrences of the action $a$ at times $t$ and $t'$, there is no action that occurs within the interval $[t+1,t'+1]$. %
This is an important expressiveness result, yet we argue that half of the picture is missing: the authors require timed words to be strictly monotone~\cite{d2007expressiveness}[Sect. 1], which impairs the expressive power of the pointwise semantics. %
To explain this, let us fall back on the general model of finite timed words, \ie consider the set $\{\rho \in \Rho\, \, |\, \, |\rho| < \omega\}$. %
Consider now the formula $\zeta_1 = \ueventually{(b\wedge \nxt{[0,0]}{a}})$ of~\cref{exincomp}. %
In the pointwise semantics, this formula captures exactly its informal description (\cref{exincomp}), \ie $\lf{\zeta_1}{pw}$  defines the language containing all finite timed words such that there exists a position $i>0$ with $\tau_i = \tau_{i+1}$, $\sigma_i = b$ and $\sigma_{i+1} = a$. %
Under $\itw$, however, the language $\lf{\zeta_1}{itw}$ is empty, and the closest language to $\lf{\zeta_1}{pw}$ is $\lf{\varphi}{itw}$ with $\varphi = \ueventually{(a\wedge b})$. %
Yet, $\lf{\varphi}{itw}$ is more permissive than $\lf{\zeta_1}{pw}$. %
Let us show this through an example. %

\begin{example}\label{exincomp2} Consider $\rho_1 = (a,0)(b,1)(a,1)(c,3.3)$ from \cref{exincomp} and  $\rho_2 = (a,0)(a,1)(b,1)(c,3.3)$. %
Then we have $\rho_1\in \lf{\zeta_1}{pw}$ and $\rho_2\notin \lf{\zeta_1}{pw}$, but both $\rho_1$ and $\rho_2$ belong to the language $\lf{\varphi}{itw}$. %
\end{example}

That is, $\zeta_1$ can distinguish between $\rho_1$ and $\rho_2$ under the pointwise semantics, whereas $\varphi$ cannot under the interval-based one. %
As we will demonstrate later, there is in fact no MTL formula that distinguishes between $\rho_1$ and $\rho_2$ under the interval-based semantics. %

\textbf{Over infinite timed words.} The same observation above continues to hold for infinite timed words as illustrated in the following example. %

\begin{example}\label{exincinf}
Consider the two infinite timed words $\rho = (\sigma,\tau)$ and $\rho' = (\sigma',\tau')$ s.t.  $^3\rho =\rho_1$, $^3\rho' =\rho_2$ and $\forall i > 3 : \sigma'_i = \sigma_i = c \wedge \tau'_i = \tau_{i} = \tau_{i-1} + 1$, with $\rho_1$ and $\rho_2$ as defined in~\cref{exincomp2}. %
That is, $\rho$ and $\rho'$ have respectively $\rho_1$ and $\rho_2$ as a prefix and agree otherwise on an infinite suffix containing only $c$ actions. %
The above observation carries on to $\rho$ and $\rho'$ as we have $\rho\in\ilf{\zeta_1}{pw}$, $\rho'\notin\ilf{\zeta_1}{pw}$ but no formula can distinguish between $\rho$ and $\rho'$ under \itw\, (this will be proven formally later on). %
\end{example}

We observe the same weakness of the interval-based semantics over timed state sequences, \ie when considering the relation $\its$ instead of $\itw$. %
To explain this, we rely on the clever remark in Bouyer et al.~\cite{bouyer2005expressiveness}[Sect. 2],~\cite{bouyer2010expressiveness}[Example 1]:

``an [infinite] timed word can be seen as a special case of timed state sequences, \eeg $\rho = (a,0)(a,1.1)(b,2)\ldots$ corresponds to the timed state sequence $\kappa = (\{a\},[0,0])(\emptyset,(0,1.1))(\{a\},$\\$[1.1,1.1])(\emptyset,(1.1,2))(\{b\},[2,2])\ldots$''

Which clearly holds under strict monotonicity. %
However, under the general model, the clearly different $\rho$ and $\rho'$ of~\cref{exincinf}  will correspond to the same timed state sequence starting with $(\{a\},[0,0])(\emptyset,(0,1))(\{a,b\},[1,1])(\emptyset,(1,3.3))(\{c\},[3.3,3.3])$. %
That is, under non-strict monotonicity of time, the encoding from infinite timed words to infinite timed state sequences, suggested in Bouyer et al.'s remark above, cannot correspond to an injective function, and therefore timed words become no longer a special case of timed state sequences. If one tries to have two different images for $\rho$ and $\rho'$, \eeg by enforcing an order for the counterpart of $\rho$, they would obtain the $\hat{\kappa} = (\{a\},[0,0])(\emptyset,(0,1))(\{b\},[1,1])(\{a\},[1,1])\ldots$, where interval adjacency is violated: $\hat{\kappa}$ is therefore \emph{not} a timed state sequence\footnote{Strictly speaking, $\hat{\kappa}$ is a \emph{weakly-monotonic} timed state sequence~\cite{alur1991logics}.}, and importantly, $\hat{\kappa}(1)$ is no longer defined.   %
This shows how the observation over finite timed words and $\itw$ above extends to finite and infinite timed words over $\itw$ and $\its$ alike. %

\subsection{Approach}\label{sec:app}

\subsubsection{Equating $\its$ and $\itw$}
In order to prove the incomparability between the pointwise and interval-based semantics, we first need to deal with the heterogeneity between timed words and timed state sequences. %
In other words, there is no way to compare $\pw$ and $\its$ if we allow arbitrary timed state sequences, due to the notable difference in nature with timed words, as explained in~\cref{sec:vs}. %
For example, the timed state sequence $(\{a,b\},[0,1))(\{c\},[1,1])$ does not make much sense in the context of action-based timed executions: actions $a$ and $b$, atomic by definition, cannot remain true continuously over the dense interval $[0,1)$. %
We need therefore to isolate the largest subset of timed state sequences that faithfully represents action-based timed executions modulo the interval-based semantics. %

We show that such subset, which we call $\Kappa_a$, can be obtained using natural properties of action-based timed executions. %
In mathematical terms, $\Kappa_a$ is the largest set of images of elements of $\Rho$ by a total relation $\mathcal{F}$ guaranteeing that $\its$ and $\itw$ coincide. %

\textbf{Action-based timed state sequences.} To construct $\Kappa_a$, we use the fact that actions occur only at punctual intervals. %
Then, since the first action occurs at $0$ and the execution, if finite, terminates with an action, the first interval is $[0,0]$ and the last interval, if exists, is punctual. %
Finally, between every two punctual intervals, we have a non-punctual one with no action associated. %
This gives us an alternation of punctual intervals with sets of actions at even positions and non-punctual ones with empty sets at odd positions. %

\begin{definition}[Action-based timed state sequences]\phantomsection
\label{def:ackk}
The set of action-based timed state sequences $\Kappa_a$ is defined as follows: 
\begin{align*}
\Kappa_a = \{\kappa = (\eta,\iota) \in \Kappa \, \, |  \, \, & \iota_0 = [0,0]  \, \, \text{ and } \\
& |\kappa| \in \naturals \Rightarrow \iota_{|\kappa|-1} \text{ is punctual } \text{ and }\\
&\forall i\in \naturals, i < |\kappa|, i\bmod 2 = 0: \eta_i\neq \emptyset \text{ and } \iota_i  \text{ is punctual } \text{and }\\
&\forall i\in \naturals, i < |\kappa|, i\bmod 2 = 1: \eta_i = \emptyset \text{ and } \iota_i  \text{ is not punctual }\}
\end{align*}
\end{definition}

\begin{example} $(\{a\},[0,0])(\emptyset,(0,1))(\{a,b\},[1,1])$ is an element of $\Kappa_a$, corresponding to an action-based execution where $a$ occurred at time $0$, and $a$ and $b$ at time $1$. %
In contrast, $(\{a\},[0,0])(\{a,b\},(0,1))(\emptyset,[1,1])$ and $(\emptyset,[0,0])(\{b\},(0,1))(\{a,b\},[1,1])$ do not represent any action-based execution and are naturally not elements of $\Kappa_a$. %
\end{example}
 
\begin{restatable}{lemma}{odd}\label{lemodd} The length of any finite element in $\Kappa_a$ is an odd number. \end{restatable}

\textbf{Relating $\Rho$ and $\Kappa_a$.} In order to define a relation between $\Rho$ and $\Kappa_a$, and to avoid heavy notations, we use an intermediary representation that allows to group events happening simultaneously while preserving their order. %
We refer to such representation, which uses the notions of \emph{flat time sequences} (\cref{def:fts}) and \emph{maximal flat subsequences} (\cref{def:mfs}), as \emph{compact timed words} (\cref{def:compact}). %

\begin{definition}[Flat time sequence]\label{def:fts}
A finite time sequence $\tau$ is said flat iff $\forall i \in \naturals, 0 < i < |\tau| : \tau_{i-1} = \tau_{i}$. 
\end{definition}

\begin{definition}[Maximal flat subsequence]\label{def:mfs}
Let $\tau$ be a time sequence and $i,j\in \naturals$ such that $i\leq j < |\tau|$. %
The finite sequence $\tau' = \tau_i\tau_{i+1}\ldots\tau_{j}$ is called a maximal flat subsequence (MFS) of $\tau$ iff $\tau'$ is a flat sequence and:
\begin{itemize}
\item{} $i > 0 \Rightarrow \tau_i > \tau_{i-1}$ and
\item{} $j+1 < |\tau| \Rightarrow \tau_{j} < \tau_{j+1}$.
\end{itemize}
\end{definition}

Note that when $i=j$ the MFS is \emph{trivial} (contains only one element). %
It is easy to see that every time sequence $\tau$ is a sequence of MFSs. %
The number of MFSs is finite if $\tau$ is finite and infinite otherwise. %
If $\tau$ is infinite, the finiteness of each MFS is guaranteed by time divergence. %
Examples will follow. %

Let $s\tau$ be the sequence over $\naturals$ defined as follows: 
\begin{itemize}
\item{} $s\tau_0 = j$ s.t. $\tau_0\ldots\tau_j$ is an MFS of $\tau$, 
\item{} $s\tau_{i+1} = k$ s.t. $\tau_{s\tau_i+1}\ldots\tau_k$  is an MFS  of $\tau$.
\end{itemize}
That is, $s\tau_i$ corresponds to the index of $\tau$ delimiting its $i^{th}$ MFS, and $|s\tau|$ corresponds to the number of MFSs in $\tau$. %
Now we are ready to define the compact version of an arbitrary timed word. %
In the following, let $\tilde{\Rho}$ denote the set of all pairs $(\tilde{\sigma},\tilde{\tau})$ where $\tilde{\sigma}$ is a (finite or infinite) sequence of non-empty finite sequences over $\Sigma$, $\tilde{\tau}$ a time sequence starting at $0$ and $|\tilde{\sigma}| = |\tilde{\tau}|$. %

  \begin{definition}[Compact timed word]\label{def:compact}
  We define the function $\mathcal{C}:\Rho\rightarrow\tilde{\Rho}$ as follows. $\mathcal{C}(\rho = (\sigma,\tau)) =  (\tilde{\rho} = (\tilde{\sigma}, \tilde{\tau}))$ iff:
  \begin{enumerate}
  \item{} $|\tilde{\rho}| = |s\tau|$,
   \item{} $\forall 0\leq  i < |\tilde{\rho}|: \tilde{\tau}_i = \tau_{s\tau_i}$,
   \item{} $\tilde{\sigma}_0 = \sigma_0\ldots \sigma_{s\tau_0}$ and $\forall 0 < i < |\tilde{\rho}|: \tilde{\sigma}_i = \sigma_{s\tau_{i-1} + 1}\ldots\sigma_{s\tau_i}$. 
  \end{enumerate}
\end{definition}

That is, $\mathcal{C}$ relates every $\rho = (\sigma,\tau) \in\Rho$ to its compact version, \ie where every $\tau_i\ldots\tau_j$ MFS of $\tau$ is reduced into a single element associated with the sequence of actions $\sigma_i\ldots\sigma_j$. 

\begin{example}\label{exnotinj} Consider $\rho_1 = (a,0)(b,1)(a,1)(c,3.3)$ and  $\rho_2 = (a,0)(a,1)(b,1)(c,3.3)$ from~\cref{exincomp2}. Then we have $\tilde{\rho}_1 = ((a),0)((b,a),1)((c),3.3)$ and $\tilde{\rho}_2 = ((a),0)((a,b),1)((c),3.3)$. 
\end{example}

\begin{restatable}{lemma}{cbij} $\mathcal{C}$ is bijective. %
\end{restatable}

\textbf{Relating $\tRho$ and $\Kappa_a$.} We now move to define a relation between compact timed words and action-based timed state sequences, where we use the function $\mathit{setof()}$ that transforms a sequence over $\Sigma$ to a set preserving its elements. %

\begin{definition}[Relation between $\tRho$ and $\Kappa_a$]\label{def:rel}
Let $\mathcal{R}:\tRho\rightarrow\Kappa_a$ be the function defined as follows. %
$\mathcal{R}(\trho = (\tsig,\ttau)) = (\kappa = (\eta,\iota))$ iff:
\begin{enumerate}
\item{}  $|\kappa| = |\trho| = \omega \text{ if $\trho$ is infinite}, |\kappa| = 2|\tilde{\rho}| - 1 \text{ otherwise and}$
\item{} $\forall i \in \{1\ldots |\tilde{\rho}| - 1\}: \eta_{2i} = \mathit{setof}(\tilde{\sigma}_i),  \, \iota_{2i} = [\tilde{\tau}_i,\tilde{\tau}_i], \, \eta_{2i-1} = \emptyset  \, \, \text{ and } \iota_{2i-1} = (\tilde{\tau}_{i-1},\tilde{\tau}_{i})  \, \, \text{ and}$
\item{}  $\eta_{0} = \mathit{setof}(\tilde{\sigma}_{0}), \iota_{0} = [\tilde{\tau}_{0},\tilde{\tau}_{0}]$.
\end{enumerate}
\end{definition}

\begin{example}\label{exinj} Consider 
$\tilde{\rho}_1 = ((a),0)((b,a),1)((c),3.3)$ and $\tilde{\rho}_2 = ((a),0)((a,b),1)((c),3.3)$ from~\cref{exnotinj}. We have $\mathcal{R}(\trho_1) = \mathcal{R}(\trho_2) = (\{a\},[0,0])(\emptyset,(0,1))(\{a,b\},[1,1])(\emptyset,(1,3.3))(\{c\},[3.3,3.3])$. %
\end{example}

In other words, the ``holes'' between the timestamps in $\trho$ are ``filled'' by the timed state sequence $\mathcal{R}(\trho)$ (\ie thanks to the non-punctual intervals between the punctual ones). %

Notice how $\mathcal{R}$ is \emph{not} injective due to the first component of a timed state sequence being a sequence of sets. %
We will show that this captures exactly the weakness of the relations $\its$ and $\itw$ alike, hinted at in~\cref{sec:int}. %

\begin{restatable}{lemma}{rsurj}\label{lem:surj} $\mathcal{R}$ is surjective. \end{restatable}

\begin{lemma} $\mathcal{R}$ is \emph{not} injective. \end{lemma}

\begin{proof} The proof is straightforward through the counterexample given in~\cref{exinj}. %
\end{proof}
 In the following, we let $\mathcal{F} = \mathcal{R}\circ\mathcal{C}$. 
 
\begin{lemma}\label{lem:tot} $\mathcal{F}$ is total. %
\end{lemma} 

\begin{proof}
$\mathcal{F}$ is a surjective function since (i) $\mathcal{C}$ is a function and (ii) $\mathcal{R}$ is a surjective function (~\cref{lem:surj}). %
Therefore $\mathcal{F}$ is right and left total and consequently total. %
\end{proof}

\textbf{\itw\, \,  and \its\, \,  are equivalent.} We are now ready to equate $\its$ and $\itw$. %
The following proposition is central as it shows that \itw\, and \its\,  enjoy exactly the same expressive power over \emph{all} timed words and \emph{all} action-based timed state sequences. %
Note that we do not need to prove the equivalence for the version with every $\kappa\in\Kappa_a$ and elements of $\mathcal{F}^{-1}[\kappa]$ since $\mathcal{F}$ is total (\cref{lem:tot}). %

\begin{restatable}{proposition}{eqitw} \label{itw=its}
For every $\phi\in\mtl$, $\rho\in \Rho$ and $t\in\rplus$ ($t\leq \mu(\rho)$ if $|\rho|\in\naturals$):
 $$\rho, t\itw \phi \Leftrightarrow \mathcal{F}(\rho),t\its\phi$$
 \end{restatable}

The main takeaway is summed by the following theorem, a direct consequence of~\cref{itw=its} and~\cref{lem:tot}, stating that $\mathcal{F}$ preserves the full expressiveness of $\its$ and $\itw$ over all timed words and all action-based timed state sequences. %

\begin{theorem}\label{thm1} For every $\phi\in\mtl$: 
\begin{itemize}
\item{} $\elf{\phi}{its} \cap \Kappa_a = \{\mathcal{F}(\rho)\, \, |\, \, \rho\in \elf{\phi}{itw}\}$ and $\elf{\phi}{itw}  = \bigcup \{\mathcal{F}^{-1}[\kappa]\, \, |\, \, \kappa\in \elf{\phi}{its}\cap \Kappa_a\}$ with $\mathit{fi\in\{f,inf\}}$,
\item{} $\alf{\phi}{its} \cap \Kappa_a = \{\mathcal{F}(\rho)\, \, |\, \, \rho\in \alf{\phi}{itw}\}$ and $\alf{\phi}{itw}  = \bigcup \{\mathcal{F}^{-1}[\kappa]\, \, |\, \, \kappa\in \alf{\phi}{its}\cap \Kappa_a\}$.
\end{itemize}  
\end{theorem}

\subsubsection{Proving Incomparability}

Since $\its$ and $\itw$ are proven equivalent, we can devise our incomparability proofs for $\pw$ \textit{vs.} $\itw$ only, while exploiting the equalities of~\cref{thm1}. %

We first formalise the notion of incomparability. %
Let $\lf{\mtl}{sem}$ denote the set of all languages over finite timed words (\ie a set of sets) defined by MTL formulae %
under the semantics corresponding to $\mathit{sem \in \{pw,itw\}}$: $\lf{\mtl}{sem} = \{\lf{\phi}{sem}\, \, |\, \,  \phi\in \mtl\}$. %
Similarly, $\ilf{\mtl}{sem} = \{\ilf{\phi}{sem}\, \, |\, \,  \phi\in \mtl\}$ and $\alf{\mtl}{sem} = \lf{\mtl}{sem} \cup \ilf{\mtl}{sem}$. %
Let \# be the binary operator defined over sets as: $S_1\#S_2 \equiv S_1\nsubseteq S_2 \wedge S_2\nsubseteq S_1$, \ie none of its arguments is included in the other. %
We can now formally define incomparability of $\pw$ and $\itw$ over MTL. 

\begin{definition}[Incomparability]\label{def:incomp}
The MTL semantics defined by $\pw$ and $\itw$ are incomparable over:
\begin{itemize}
\item{} Finite timed words iff $\lf{\mtl}{\pww} \#\, \,  \lf{\mtl}{\itww}$, 
\item{} Infinite timed words iff  $\ilf{\mtl}{\pww} \#\, \,  \ilf{\mtl}{\itww}$. 
\end{itemize}
The semantics defined by $\pw$ and $\itw$ are incomparable iff they are incomparable over finite and infinite timed words, \ie  $\alf{\mtl}{\pww} \#\, \,  \alf{\mtl}{\itww}$.
\end{definition}

In order to show that \eeg $\lf{\mtl}{pw} \# \lf{\mtl}{itw}$, it suffices to find two formulae $\phi_1,\phi_2\in \mtl$ such that for all $\psi\in\mtl$, the language of finite timed words defined by $\phi_1$ (\resp $\phi_2$) under $\pw$ (\resp $\itw$) is different from the language defined by $\psi$ under $\itw$ (\resp $\pw$) as generalised by the following proposition. %

\begin{restatable}{proposition}{propinc}\label{prop:comp}
$\lf{\mtl}{\pww} \# \lf{\mtl}{\itww}$ iff $\exists\phi_1,\phi_2\in\mtl\, \, \forall\psi\in\mtl$:
\begin{itemize}
\item $\lf{\phi_1}{\pww} \neq \lf{\psi}{\itww}$ and
\item  $\lf{\phi_2}{\itww} \neq \lf{\psi}{\pww}$. 
\end{itemize}
(The equivalence with $\ilf{\mtl}{\pww} \# \ilf{\mtl}{\itww}$ is defined analogously). %
\end{restatable}

\textbf{Proving $\lf{\mtl}{pw} \# \lf{\mtl}{itw}$.} %
We start with proving $\lf{\mtl}{pw} \nsubseteq \lf{\mtl}{itw}$. %
It suffices to find two distinct finite timed words that have the same image by $\mathcal{F}$ and a formula $\phi$ that distinguishes between them under $\pw$. %
By definition, no formula under $\its$ can distinguish between their images through $\mathcal{F}$ (since it is the same image) and we conclude that no formula can distinguish between $\rho_1$ and $\rho_2$ under $\itw$ since $\its$ and $\itw$ are equally expressive. %

\begin{restatable}{lemma}{pwf}\label{lem:pwf} Let $\zeta_1 = \ueventually{(b \wedge \nxt{[0,0]}{a})}$ from~\cref{exincomp}. %
The following statement holds: $$\forall\psi\in\mtl: \lf{\zeta_1}{pw} \neq \lf{\psi}{itw}$$
\end{restatable}

\begin{proof}
Let $\rho_1 = (a,0)(b,1)(a,1)(c,3.3)$ and $\rho_2 = (a,0)(a,1)(b,1)(c,3.3)$ from~\cref{exincomp2}. %
We have $\rho_1\in\lf{\zeta_1}{pw}$ and $\rho_2\notin\lf{\zeta_1}{pw}$. 
Now, $\rho_1$ and $\rho_2$ have the same image by $\mathcal{F}$ that we will refer to as $\kappa$. %
By definition, for every $\psi\in\mtl$, we have either $\kappa\in\lf{\psi}{its}$, and by~\cref{thm1}, $\rho_1\in\lf{\psi}{itw}$ and $\rho_2\in\lf{\psi}{itw}$, \emph{or} $\kappa\notin\lf{\psi}{its}$,  and again by~\cref{thm1} $\rho_1\notin\lf{\psi}{itw}$ and $\rho_2\notin\lf{\psi}{itw}$. And therefore $\lf{\zeta_1}{pw}\neq \lf{\psi}{itw}$. 
\end{proof}

\begin{remark}\label{rem:ltl} Note that this weakness of the interval-based semantics has a direct consequence on the coherence between MTL and its untimed version, LTL. %
Consider \eeg $\rho_1$ and $\rho_2$ above. %
Their untimed counterparts are, respectively, the words $abac$ and $aabc$ which are temporally different (in terms of LTL). Yet, under $\itw$ or $\its$ alike, it is impossible to distinguish between them whenever the two letters in the middle occur simultaneously.
\end{remark} 

For the part $\lf{\mtl}{itw} \nsubseteq \lf{\mtl}{pw}$, we depart from

$$\phi_{\mathit{dp}} = \neg\ueventually{(a\wedge(\uuntil{\neg (\varphi_\Sigma}{a}) \wedge (\uuntil{\neg a}{(\neg a \wedge \eventually{[1,1]}{(a \vee b)})}))}$$

The formula proven to have no equivalent under $\pw$ over strictly monotone finite timed words~\cite{d2007expressiveness} (\cref{sec:int}). %
The idea here is rather simple: find a formula $\gamma_1$ that captures exactly the language of $\phi_{\mathit{dp}}$ under the general model of finite timed words then conclude. %

\begin{restatable}{lemma}{itwf}\label{lem:itwf} Let $\gamma_1 = \phi_{\mathit{dp}} \wedge \phi_s$ with\footnote{The first clause of $\phi_s$ is necessary due to the strict nature of the Until operator, from which \uglobally{} is derived.} $\phi_s = (\neg \bigvee_{a\neq b\in\Sigma} (a\wedge b)) \wedge \uglobally (\neg \bigvee_{a\neq b\in\Sigma} (a\wedge b))$ \\
The following statement holds: $\forall\psi\in\mtl: \lf{\gamma_1}{itw} \neq \lf{\psi}{pw}$
\end{restatable} 

\begin{proof} 
The language $\lf{\gamma_1}{itw}$ contains the set of strictly monotone (satisfying $\phi_s$) finite timed words satisfying $\phi_{\mathit{dp}}$. %
The proof then follows directly from~\cite{d2007expressiveness}[Sect. 3].
\end{proof}

\begin{proposition}\label{prop:incf} The pointwise and the interval based semantics are incomparable over finite timed words.  \end{proposition}
\begin{proof} From~\cref{lem:pwf}, \cref{lem:itwf} and~\cref{prop:comp} we get $\lf{\mtl}{pw} \# \lf{\mtl}{itw}$ and conclude by~\cref{def:incomp}. \end{proof}

\textbf{Proving $\ilf{\mtl}{pw} \# \ilf{\mtl}{itw}$.} We first prove $\lf{\mtl}{pw} \nsubseteq \lf{\mtl}{itw}$, using the same formula from~\cref{lem:pwf}. %

\begin{lemma}\label{lem:ipw} Let $\zeta_1 = \ueventually{(b \wedge \nxt{[0,0]}{a})}$. The following statement holds: $$\forall\psi\in\mtl: \ilf{\zeta_1}{pw} \neq \ilf{\psi}{itw}$$
 \end{lemma}

\begin{proof} The proof is very similar to that of~\cref{lem:pwf}, with $\rho\in\ilf{\zeta_1}{pw}$ and $\rho'\notin\ilf{\zeta_1}{pw}$ (\cref{exincinf}) having the same image by $\mathcal{F}$. 
\end{proof} 

For $\ilf{\mtl}{itw} \nsubseteq \ilf{\mtl}{pw}$, we use a formula already proven in~\cite{bouyer2005expressiveness}[Prop. 1 and Sect. 3.2] to have a language over infinite words that is not expressible under $\pw$. %
This language is informally defined as ``all infinite timed words containing an occurrence of $b$ within the absolute interval $(0,2)$ and an occurrence of  $c$ within the absolute interval $(1,2]$, such that the latter happens \emph{after} the former within the relative interval $(0,1)$''.

\begin{lemma}\label{lem:itw} Let $\gamma_2 = \eventually{[0,1]}{(\eventually{[0,1)}{b} \wedge \eventually{[1,1]}{c})}$. 
The following statement holds: $$\forall\psi\in\mtl: \ilf{\gamma_2}{itw} \neq \ilf{\psi}{pw}$$
\end{lemma}

\begin{proposition}\label{prop:incinf} The pointwise and the interval-based semantics are incomparable over infinite timed words. \end{proposition}

\begin{proof} From~\cref{lem:ipw}, \cref{lem:itw} and~\cref{prop:comp} we get $\ilf{\mtl}{pw} \# \ilf{\mtl}{itw}$ and conclude by~\cref{def:incomp}. \end{proof}

And finally we conclude with the following theorem, a direct consequence of~\cref{prop:incf} and~\cref{prop:incinf}. 
 
\begin{theorem}\label{thm:inc} The pointwise and the interval-based semantics of MTL are incomparable. 
\end{theorem}

\textbf{Incomparability under stutter freeness.} Our incomparability results continue to hold under the reasonable restriction of stutter freeness (\cref{def:stut}, \cref{def:stut1}). %
For the finite case, the formula $\zeta_1$ given in~\cref{lem:pwf}, as shown in the proof of the same lemma,  distinguishes between two stutter-free finite timed words under $\pw$ whereas no formula under $\itw$ can do so (and the same goes for the infinite case, ~\cref{lem:ipw} and its proof). %
The only technical difference when disallowing stuttering is the fact that a finite timed state sequence, element of $\Kappa_a$, will then have an image of finite cardinality by the inverse relation $\mathcal{F}^{-1}$, %
For example, given $\kappa = (\{a\},[0,0])(\emptyset,(0,1))(\{a,b\},[1,1])$, its image by $\mathcal{F}^{-1}$ contains an infinity of finite timed words when stuttering is allowed $(a,0)(a,0)(a,1)(b,1)(b,1)$, $(a,0)(a,0)(a,1)(b,1)(a,1)(b,1)$ etc. but only two finite timed words under stutter freeness, \ie $(a,0)(a,1)(b,1)$ and $(a,0)(b,1)(a,1)$. %
In summary, stutter freeness
which has no effect on the incomparability theorem. %

%% file: 4.monitors.tex
\section{Mixed Semantics}%
\label{sec:monitors}

In this section, we present a new \emph{mixed} semantics that embeds both the pointwise and the interval-based ones. %
We do so in three steps. %
First, we present an intuitive version of the mixed semantics over MTL that,  leveraging the expressive power of the interval-based semantics, is strictly more expressive than the pointwise one (\cref{sec:du}). %
Second, in~\cref{sec:mixint}, we conjecture that such version is incomparable with the interval-based one (due mainly to the fact that the latter is more suited for signals). 
Third, we show that extending the syntax of MTL with a single atomic formula unleashes the power of the mixed semantics leading to the embedding of both the pointwise and the interval-based semantics, conjugating therefore their full expressive powers (\cref{sec:ext}). %
The embedding is automated through syntax-driven compilers. %

\subsection{Dual reasoning on time and positions}\label{sec:du} 
In this subsection, we undertake the first step described above, \ie we propose to increase the power of the pointwise semantics using the flexibility of the interval-based one. %
The reason we do it in this direction (rather than the other way around) is the fact that the interval-based semantics is more suitable for signals (\cref{sec:vs}) whereas we focus  on action-based timed executions (another example on the counter-intuitiveness of the interval-based semantics over timed words will be given in~\cref{sec:mixint}). 
We still provide in~\cref{sec:ext} an extended version of the mixed semantics that embeds both the pointwise and the interval-based ones. 

Given a compact timed word $\trho = (\tsig,\ttau)$, we define its duration in the obvious way, \ie $\mu(\trho) = \mu(\mathcal{C}^{-1}(\trho))$. %
For any $t\in\rplus$ (with $t\leq\mu(\trho)$ if $\trho$ is finite), we let $\trho(t)$ equal $\tsig_i$ where $\ttau_i = t$ (such $\ttau_i$ is uniquely defined if it exists) and the singleton sequence $(\vdash)$ otherwise, where  $\vdash$ is a ``helper'' symbole that represents ``no action'', \ie it will be semantically nothing else than equivalent to $\neg\varphi_{\Sigma}$. %

The idea is to evaluate an MTL formula over a compact timed word $\trho = (\tsig,\ttau)$, a time point $t\in\rplus$ (with $t\leq\mu(\trho)$ if $|\trho|<\omega$)  and a position $0\leq j < |\trho(t)|$. %
It is important to remark here that the condition  $0\leq j < |\trho(t)|$ is by no means impractical for verification purposes, as the satisfaction of some formula $\phi$ by some $\trho$ boils down to the satisfaction of $\phi$ by $\trho$, $t_0 = 0$ and $j_0 = 0$ (examples will follow). %

Now we are ready to formalise the mixed semantics: 

   $$
  \begin{array}{lll}
      \trho,t,j\mx a &\text{ iff }& a = \trho(t)(j)
      \\
      \trho,t,j\mx \neg \phi &\text{ iff }& \trho,t,j\nmx \phi
      \\
      \trho,t,j\mx \phi \wedge \psi &\text{ iff }& \trho,t,j\mx \phi \text{ and } \trho,t,j\mx \psi
      \\
      \trho,t,j\mx \until{\phi}{I}{\psi} &\text{ iff} & \exists (t,j) < (t',j')  : t'-t\in I \text{ and } \trho,t',j'\mx \psi \text{ and } \\
      &&\forall (t,j) < (t'',j'') < (t',j') : \trho,t'',j''\mx \phi%
    \end{array}
$$

Where the relation $<$, used for the $\until{}{I}{}$ operator, denotes a lexicographic  order defined in the usual way. %
We say $\trho$ satisfies $\phi$, written $\trho \mx \phi$, iff $\trho,0,0\mx\phi$. %

Let us explain, relying on the rules of the semantics above, how the relation $\mx$ increases the expressive power of $\pw$ (we will formally prove this soon).  %
The real $t$ enables the evaluation at any arbitrary time point, whereas the natural $j$ allows to ``access'' the positions of ordered events occurring at $t$, if any. %
The semantics of $\until{}{I}{}$ reflects this increased power of expressiveness: $\until{\phi}{I}{\psi}$ is evaluated at any timepoint and any position of the sequence of events occurring at $t$. %
Such position is guaranteed to exist: in particular, if $t$ belongs to a ``hole'' in $\trho$ (\ie strictly comprised between two timestamps), then $\trho(t)(0)$ exists and is equal to the $\vdash$ symbole which plays otherwise a transparent role as, in such case, $\trho,t,0\models \neg\varphi_\Sigma$ and $\trho,t,0\models \top$ as expected. %
The following example illustrates the intuitive nature of the mixed semantics. %

\begin{example}\label{exmx} Consider the setting of~\cref{exincomp} with $\rho_1 = (a,0)(b,1)(a,1)(c,3.3)$, $\zeta_1 = \ueventually{(b \wedge \nxt{[0,0]}{a})}$ and $\zeta_2 = \eventually{(0,1)}{\eventually{[0,3.5]}{c}}$. %
As explained within the same example, $\rho_1$ should, intuitively, satisfy both formulae, yet we have $\rho_1 \nvDash_{itw} \zeta_1$ and $\rho_1 \nvDash_{pw} \zeta_2$. 

Under the mixed semantics, both formulae are satisfied thanks to the flexibility of $\mx$, taking both time points and positions into account. %
To illustrate this, let us first encode in the obvious way $\rho_1 = (a,0)(b,1)(a,1)(c,3.3)$ into its compact counterpart $\trho_1 = ((a),0)((b,a),1)((c),3.3)$. %

Then we have $\trho_1 \mx \zeta_1$ since $\trho_1,0,0 \mx \zeta_1$: we have $\trho_1,t',j'$ with $(t' = 1, j' = 0)$, greater then $(0,0)$ in lexicographic order, satisfies $b$ (since $\trho_1(1)(0) = b$) and in turn, $\trho_1,t'',j''$ with $(t'' = 1,j'' = 1)$, greater than $(1,0)$ satisfies $a$ within $[0,0]$ and (i) for every $(0,0)<(t_1,j_1)<(t', j')$, $\trho_1,t_1,j_1 \mx\top$ and (ii) for every $(t',j')<(t_2,j_2)<(t'', j'')$, $\trho_1,t_2,j_2 \mx\neg\varphi_\Sigma$ (vacuous truth). %

Dually, $\trho_1,0,0 \mx \zeta_2$ because for any couple $(t',j'=0)$ with $t'\in(0,1)$, we have $\trho,t',j'\mx \eventually{[0,3.5]}{c}\, $  since at $(t''=3.3 ,j'' = 0)$ we have $\trho_1,t'',j''\mx c$ (since $\trho_1(t'')(j'')=c$) and $t'' - t' \in [0,3.5]$. %
\end{example}

\subsubsection{\mx\, is strictly more expressive than \pw} 
The mixed semantics is strictly more expressive than the pointwise one. %
To show this, we proceed as follows. %
First, we devise a compiler $\pcompile{-}$ over MTL such that  $\pcompile{\phi}$ under $\mx$ is equivalent to $\phi$ under $\pw$ (using the bijective function $\mathcal{C}$ to move between timed words and their compact version). %
Second, we give an example of a formula not expressible in the pointwise semantics (from~\cref{sec:app}) and show that it is expressible in the mixed semantics. %

\textbf{From \pw\,  to \mx.}
We reason following the main observation that, in $\pw$, formulae are evaluated only at timepoints where actions occur. %
Thus, $\pcompile{\phi}$ must satisfy $\varphi_\Sigma$, regardless of its argument. %
The case for $\until{}{I}{}$ is the involved one. %
Given a formula $\until{\phi}{I}{\psi}$ interpreted under $\pw$, to get its equivalent formula under $\mx$, we need also to ensure that (i) $\psi$ holds at  a timepoint corresponding to an action occurrence and (ii) $\phi$ does not need to hold at timepoints where no action occurs, as formalised by $\pcompile{\until{\phi}{I}{\psi}}$ in the following definition. %

\begin{definition}[$\pw$-preserving compilation]
The compiler $\pcompile{-}$ is defined inductively as follows: 

\begin{itemize}
\item{} $\pcompile{a} = a$, 
\item{} $\pcompile{\neg \phi} = \varphi_{\Sigma} \wedge \neg \pcompile{\phi}$,
\item{} $\pcompile{\phi\wedge\psi} =  \varphi_{\Sigma} \wedge \pcompile{\phi} \wedge \pcompile{\psi}$,
\item{} $\pcompile{\until{\phi}{I}{\psi}} = \varphi_\Sigma \wedge \until{(\pcompile{\phi}\vee\neg\varphi_\Sigma)}{I}{(\pcompile{\psi}\wedge\varphi_\Sigma)}$.
\end{itemize}
\end{definition} 

\begin{restatable}{proposition}{compileits}\label{prop:cits}

For every $\phi\in\mtl$ and $\rho\in \Rho: \rho\pw\phi \Leftrightarrow \mathcal{C}(\rho)\mx\pcompile{\phi}$. %

\end{restatable} 

\textbf{A language definable under $\mx$ but not under $\pw$.}
Consider the formula $\gamma_2 = \eventually{[0,1]}{(\eventually{[0,1)}{b} \wedge \eventually{[1,1]}{c})}$ from~\cref{lem:itw}. %
As states the same lemma, $\gamma_2$ has no equivalent under the pointwise semantics over infinite timed words. %
The following lemma states that $\gamma_2$ does have an equivalent under \mx\, over infinite timed words. %

\begin{restatable}{lemma}{supmxpw}\label{lem:exist} Let $\gamma_2 = \eventually{[0,1]}{(\eventually{[0,1)}{b} \wedge \eventually{[1,1]}{c})}$ and $\gamma_3 = \eventually{(0,1]}{(\eventually{[0,1)}{b} \wedge \eventually{[1,1]}{c}})$. Then:
 $$\ilf{\gamma_2}{itw} = \{\mathcal{C}^{-1}(\trho)\, \, | \, \, \trho\in\ilf{\gamma_3}{mx}\}$$
\end{restatable}

\begin{remark}
Note the subtlety in $\gamma_3$ where the interval of the outer eventually operator is left-open at $0$. %
Under $\itw$, $\gamma_2$ and $\gamma_3$ are equivalent because the time point satisfying the formula must be greater than $0$ by the definition of the (strict) Until operator from which $\eventually{[0,1]}{}$ is derived. %
Under $\mx$, on the other hand, we have $\ilf{\gamma_3}{mx} \subsetneq \ilf{\gamma_2}{mx}$   since \eeg $(t=0,j'=2)$ is greater than $(t=0,j=0)$ (but $t$ is not greater than $0$); we have therefore the compact version of any infinite timed word starting with $(a,0)(c,0)(b,0)(c,1)(c,3)$ in $\ilf{\gamma_2}{mx}$ but not in $\ilf{\gamma_3}{mx}$. %
Hence the explicit exclusion of $0$ from the interval. %
\end{remark} 

\begin{theorem}\label{thm:smxpw} The mixed semantics is more expressive than the poitnwise one. %
\end{theorem}

\begin{proof} Direct from~\cref{lem:exist} and~\cref{prop:cits}. 
\end{proof}

\subsection{Mixed semantics vs. interval-based semantics.}\label{sec:mixint}
We now explain the fundamental difference between \itw\, and our \mx\, and why we argue that \mx\, is more suitable for action-based executions. 
Let $\phi = \until{(a\vee\neg\varphi_\Sigma)}{[1,2]}{b}$ and $\rho = (a,0)(a,0.5)(c,0.5)(c,1.5)(b,1.5)$. %
The (action-based) intuition would suggest that $\rho$ should not satisfy $\phi$, two occurrences of $c$ at different dates take place (which correspond to neither an $a$ nor ``no action'') before $b$ does, and indeed $\trho$ does not satisfy $\phi$ under the mixed semantics. %
In contrast, we have $\rho\itw\phi$, since the interval-based semantics is more suitable for signals. %
If $\rho$ was a signal, 
it would naturally make sense for it to satisfy $\phi$ since at every time point, either the proposition $a$ held or no proposition did until $b$ within $[1,2]$. %

Because of this ``permissive'' nature of the $\until{}{I}{}$ operator at its first argument under \itw (and its more ``restrictive'' nature at its second argument, see the second example below), we conjecture that the satisfaction relations \mx\, and \itw\, are incomparable over timed words. %
The proof of this conjecture is already halfway since we know that the $\mx$ is strictly more expressive than $\pw$ (\cref{thm:smxpw}), and given that $\itw$ and $\pw$ are incomparable (\cref{thm:inc}), $\mx$ is either more expressive than or incomparable with $\itw$. %
Actually, we can relatively easily construct an MTL formula defining a language under $\mx$ (after un-compacting with $\mathcal{C}^{-1}$) that cannot be characterised by any MTL formula under $\itw$ or $\pw$ alike, as states the following lemma.  
\begin{restatable}{lemma}{supall}\label{lem:supall}
Let $\gamma_3 = \eventually{(0,1]}{(\eventually{[0,1)}{b} \wedge \eventually{[1,1]}{c})}$ from~\cref{lem:exist}, $\zeta_1 = \ueventually{(b \wedge \nxt{[0,0]}{a})}$ from~\cref{lem:pwf}, $\gamma_4 = \globally{(0,2]}{\neg\varphi_\Sigma}$ and $\gamma_5 = \gamma_3 \vee (\zeta_1 \wedge \gamma_4)$. %
Then for all $\phi\in\mtl$: 
\begin{itemize}
\item{} $\ilf{\phi}{\pww} \neq \{\mathcal{C}^{-1}(\trho)\, \, |\, \, \trho\in\ilf{\gamma_5}{\mxx}\}$
\item{} $\ilf{\phi}{\itww} \neq \{\mathcal{C}^{-1}(\rho)\, \, |\, \, \trho\in\ilf{\gamma_5}{\mxx}\}$
\end{itemize}
\end{restatable}

To have a less biased view, let us now look at a trickier example that rather pleads in favour of \itw. %
Let $\rho = (c,0)(c,0.5)(c,1.5)(b,1.5)$ and $\phi = \until{(c\vee\neg\varphi_\Sigma)}{[1,2]}{(b\wedge \neg c)}$. %
Here, we have $\trho\mx \phi$ but $\rho\nitw\phi$. %
This is because $b\wedge \neg c$ is trivially satisfied at $(t,j) = (1.5,1)$ under \mx, whereas it is not at $t = 1.5$ under $\itw$. %
It is however legitimate for a user to mean by this formula ``verify that only occurrences of $c$ are seen until an occurrence of $b$ takes place between times $1$ and $2$  \emph{and} $c$ does not occur simultaneously with such $b$''. %
The mixed semantics fails to capture this, and at first glance the naive way to remedy this issue would be to use an immediate past operator that allows to look \emph{behind} $(t,j) = (1.5,1)$ while maintaining the same timepoint, \ie at $(1.5,0)$. %
In the next subsection, we show that this ``immediate past'' is an overkill as we can instead use a single atomic formula to extend the syntax of MTL in order to strictly embed the full expressive power of \itw\, into \mx. %

\subsection{Extending mixed semantics}\label{sec:ext}
Following the examples given in~\cref{sec:mixint}, and in order to embed \itw, the satisfaction relation \mx\,  needs to be able to force the evaluation of a formula at index $0$ within any $\trho(t)$, looking forward at the whole sequence returned by $\trho(t)$ (instead of complexifying the syntax of MTL, and \mx\,  thereof, with an immediate past operator). %
Let us consider the example $\until{(a\vee\neg\varphi_\Sigma)}{I}{b}$ satisfied by some $\rho,t$ under \itw, and let $t$ correspond to no timestamp of $\rho$ to simplify the explanation. %
To satisfy the same formula at $\trho,t,0$ under \mx,  intuitively we need to check for $(a\vee \eventually{[0,0]}{a})$ only at position $0$ for every $\trho(t') \neq (\vdash), t'>t$ until $(b\vee \eventually{[0,0]}{b})$ becomes true at position $0$ of some $\trho(t'')$ satisfying $t''- t\in I$. %

To do this in a simple and efficient way, the idea is to interpret the mixed semantics over $\Sigma\cup\{\beta\}$, where $\beta$ is an atomic proposition not in $\Sigma$. The syntax of MTL is accordingly extended with the atomic formula $\beta$, we call such extended syntax MTL$_{\beta}$. The semantics of $\beta$ is defined simply to hold only at the beginning of a sequence $\trho(t)$: $$\trho,t,j\mx \beta \text{ iff } j = 0$$ %

This seemingly minuscule addition comes with a major expressiveness gain. %
The mixed semantics now embeds the interval-based one using the compiler $\icompile{-} : \mathit{MTL\rightarrow MTL_{\beta}}$ defined inductively as follows: 
\begin{itemize}
\item{} $\icompile{a} = a \vee \eventually{[0,0]}{a}$, 
\item{} $\icompile{\neg\phi} = \neg \icompile{\phi}$,
\item{} $\icompile{\phi\wedge\psi} = \icompile{\phi}\wedge\icompile{\psi}$,
\item{} $\icompile{\until{\phi}{I}{\psi}} = 
{\until{(\beta\Rightarrow\icompile{\phi})}{I}{(\beta\wedge\icompile{\psi})}}$. 
\end{itemize}

The following proposition states that \mx\, over MTL$_{\beta}$ embeds \itw\, over MTL. 

\begin{restatable}{proposition}{compileitw}\label{prop:itw}
For all $\phi$ in MTL and $\rho\in\Rho$: $\rho\itw\phi \Leftrightarrow \mathcal{C}(\rho)\mx \icompile{\phi}$.
\end{restatable}

Strict expressiveness superiority follows directly from~\cref{prop:itw} and~\cref{thm:smxpw}. %

\begin{remark} 
The embedding of \itw\, into $\mx$ (\cref{prop:itw}) looks asymmetrical since we allow an extended syntax for \mx\, but not for \itw. %
However, such asymmetry is a consequence of simplifying the presentation rather than anything else. Adding a straightforward interpretation of $\beta$ under \itw, \ie $\rho,t\itw\beta \text{ iff } t = 0$ will not change anything: since we have $\rho\itw\phi$ iff $\rho,0\itw\phi$, we get $\rho\itw\beta$ iff $\rho,0\itw\beta$, and therefore $\alf{\beta}{itw} = \Rho = \alf{\top}{itw}$. %
In other words, the expressive power of \itw\,  over MTL and MTL$_{\beta}$ would be identical. %
\end{remark} 

We conclude with showing that the set of languages characterised by $\mx$ over MTL$_{\beta}$ is strictly larger than that characterised by the union of the two sets of languages characterised by $\itw$ and $\pw$ over MTL.

\begin{theorem}\label{thm:last}  $$\alf{\mtl}{\pww} \cup \alf{\mtl}{\itww} \subsetneq \{\{\mathcal{C}^{-1}(\tilde{\rho})\, \, | \, \, \tilde{\rho}\in\alf{\phi}{\mxx}\}\, \, |\, \, \phi\in\mtl_{\beta}\}$$ 
\end{theorem}
\begin{proof} Direct from~\cref{prop:cits}, \cref{prop:itw} and~\cref{lem:supall}. %
\end{proof}
Note that the witness formula for the strictness of the inclusion in~\cref{thm:last} is $\gamma_5$ in~\cref{lem:supall}, which does not use the atomic formula $\beta$. This is an interesting observation: $\mx$ over MTL already characterises at least one language that can be expressed neither in $\itw$ nor in $\pw$.

%% file: 8.conclusion.tex
\section{Conclusion}
\label{sec:conclusion}

In this paper, we revisited an important aspect of the extensively studied expressiveness of MTL, namely that of its pointwise semantics as opposed to the interval-based one. %
We have demonstrated that the two semantics are incomparable over the general models of action-based timed executions, giving thereby a clear answer, by the negative, to the question on whether reasoning over arbitrary time points \emph{instead of} ordered positions of timed events increases the expressive power of MTL. %
To get the better of both worlds over action-based executions, we proposed a new mixed semantics that embeds the interval-based and the pointwise ones, therefore conjugating their full expressive powers. %

For future work, the results presented within this paper call for further research on at least two fronts. 
On the decidability/complexity front, model checking is already shown to be undecidable over infinite timed words for both the interval-based and the pointwise semantics~\cite{ouaknine2006metric,alur1991logics} and over finite timed words for the interval-based one~\cite{alur1991logics}. %
An interesting future direction would be to rather investigate the complexity of monitoring under the mixed semantics, already shown to be tractable for the pointwise semantics \eeg in~\cite{ho2014online}. %
In particular, offline monitoring is a promising venue to reason on executions of distributed systems where the order over simultaneous events is a key aspect~\cite{liebig1999event,fidge2002logical,lamport2019time}. %
On the expressiveness front, it would be interesting to consider extending our mixed semantics to reason on mixed signal-event executions while maintaining reasonable complexity. %
This would require to define a mixed structure, \eeg $((a,b),[0,0])(\{p,q\},(0,1))((c),[1,1])$ where $a$, $b$ and $c$ are in the actions set and $p$,$q$ in the atomic propositions one, then further extend our mixed semantics to both sets. %
In the same vein, we aim to investigate relating the mixed semantics to the different semantics of more expressive logics such as TPTL~\cite{alur1994really} and \logic~\cite{DBLP:conf/ecoop/Amara0FF25}. %

%% file: incomparability.tex
\section{Proofs about incomparability}\label{app:inc}

\odd*

\begin{proof} Let $\kappa = (\eta,\iota)$ be an element of $\Kappa_a$ such that $|\kappa|\in\naturals$. %
Suppose that $|\kappa|\bmod 2 = 0$. %
This means that $\iota_{|\kappa| - 1}$ is not punctual which contradicts~\cref{def:ackk}. 
\end{proof}

\cbij*

\begin{proof} 

\emph{$\mathcal{C}$ is injective.} To prove injection, suppose that two different timed words $\rho = (\sigma,\tau)$ and $\rho' = (\sigma',\tau')$ have the same image $\trho = (\tsig,\ttau)$ by $\mathcal{C}$. %
If $|\rho| \neq |\rho'|$, then we have $|\tilde{\rho}| = |s\tau| = |s\tau'|$ leading to a contradiction. %
Suppose now $|\rho| = |\rho'|$. %
Since by~\cref{def:compact} we have $|\tilde{\rho}| = |s\tau| = |s\tau'|$, we get from the same definition necessarily $\forall 0\leq  i < |\tilde{\rho}|: \tilde{\tau}_i = \tau_{s\tau_i} = \tau_{s\tau'_i}$, which means that $\tau$ and $\tau'$ contain exactly the same sequences of MFSs and therefore $\tau = \tau'$. %
Now, since $\rho$ and $\rho'$ are different, there exists $i < |\rho|$ such that $\rho$ and $\rho'$ differ, and since $\tau=\tau'$, the only remaining possibility is $\sigma_i \neq \sigma'_i$. %
Since by~\cref{def:compact} we have $0 \leq i \leq s\tau_0$ or $\exists j,k: s\tau_j<i\leq s\tau_k$, this leads to $\tsig_0$ or $\tsig_k$ equalling a sequence differing at one element. Contradiction. %

\emph{$\mathcal{C}$ is surjective.} We prove that for every $\trho = (\tsig,\ttau) \in\tRho$, $\mathcal{C}^{-1}[\trho]$ is not empty. %
Whether $\trho$ is finite or infinite, there exists by definition a sequence over $\naturals$ with $|s\tau| = |\trho|$ and   $\tau_{s\tau_i} = \tilde{\tau}_i$ for every $0\leq i < |s\tau|$. %
Accordingly, there exists a time sequence $\tau$ corresponding to $s\tau$ and there exists $\rho = (\sigma,\tau)$ satisfying $\sigma_0\ldots s\tau_0 = \tsig_0$ and $\sigma_{s\tau_{i-1} + 1}\ldots\sigma_{s\tau_i} = \tsig_i$ for all  $0 < i < |s\tau|$. %

\end{proof}

\rsurj*

\begin{proof}
To prove that $\mathcal{R}$ is surjective, it suffices to prove that the inverse relation $\mathcal{R}^{-1}$ is left total. %
Let $\kappa$ be an arbitrary finite element of $\Kappa_a$. %
We want to prove that there exists some $\tilde{\rho}$ (not unique in this case) satisfying the conditions given by $\mathcal{R}$. %
If $\kappa$ is finite, any such $\tilde{\rho}$ must be of length $\frac{|\kappa| + 1}{2}$, which is guaranteed to be a natural since $|\kappa|$ is odd (\cref{lemodd}). %
We can construct an instance of such $\tilde{\rho}$ from $\kappa$ as follows: (1) $\tilde{\sigma}_{i}$ as a sequence in \eeg alphabetical order from $\eta_{2i}$, (2) $\tilde{\tau}_i = (\iota^l_{2i})$ for every $i\in\{0\ldots\frac{|\kappa| + 1}{2} - 1\}$. %
If $\kappa$ is infinite, then we can relate $\kappa$ to \eeg an infinite compact timed words $\trho$ satisfying conditions (1) and (2) above for every $i\in\naturals$. 
  
\end{proof}

\eqitw*

\begin{proof}

By structural induction on $\phi$. %

\begin{itemize}
\item{} Case $a$: ($\Rightarrow$) By definition, $\rho, t\models_{itw} a$ iff $\exists i : \sigma_i = a$ and $\tau_i = t$. %
By the function $\mathcal{C}$, we have $\exists j: a\in \tilde{\sigma}_j$ and $\tilde{\tau}_j = t$. %
By the function $\mathcal{R}$, we have $a\in \eta_{2j}$ and $t\in \iota_{2j} = [t,t]$, which implies that $a\in(\mathcal{F}(\rho))(t)$ and therefore $\mathcal{F}(\rho),t\models_{its} a$. \\%
($\Leftarrow$) By definition, $\mathcal{F}(\rho),t\models_{its} a$ iff $a\in\mathcal{F}(\rho)(t)$, which implies the existence of a position $i$ such that $a\in\eta_{i}$ and $t\in\iota_i$. %
By the definition of $\Kappa_a$ (\cref{def:ackk}), such $i$ is  even with $\iota_i$ punctual, and therefore $\iota_i = [t,t]$. %
By the relation $\mathcal{R}^{-1}$, we have $a\in\tsig_{\frac{i}{2}}$ and $\ttau_{\frac{i}{2}} = t$ and finally by the function $\mathcal{C}^{-1}$ there exists $j$ such that $\tau_j = t$ and $\sigma_j = a$. %

\item{} Case $\neg \phi$: direct from the induction hypothesis $\rho, t\itw \phi \Leftrightarrow \mathcal{F}(\rho),t\its\phi$. %

\item{} Case $\phi\wedge\psi$. %
By definition of \itw:
$$\rho, t\itw \phi\wedge\psi \Leftrightarrow \rho,t\itw\phi \text{ and } \rho,t\itw\psi$$
By the induction hypotheses $\rho, t\itw \phi \Leftrightarrow \mathcal{F}(\rho),t\its\phi$ and $\rho, t\itw \psi \Leftrightarrow \mathcal{F}(\rho),t\its\psi$ we have: 
$$\rho, t\itw \phi\wedge\psi \Leftrightarrow \mathcal{F}(\rho),t\its\phi \text{ and } \mathcal{F}(\rho),t\its\psi$$ 
and finally by definition of \its:
$$\rho, t\itw \phi\wedge\psi \Leftrightarrow \mathcal{F}(\rho),t\its\phi \wedge\psi$$

\item{} Case $\until{\phi}{I}{\psi}$.

By definition of \itw:
$$
  \begin{array}{lll}
        \rho,t\itw \until{\phi}{I}{\psi} & \Leftrightarrow & \exists t' > t  : t'-t\in I \text{ and } \rho,t'\itw \psi \text{ and } \forall t < t'' < t' : \rho,t''\itw \phi%
    \end{array}
$$

By the induction hypotheses $\rho, t''\itw \phi \Leftrightarrow \mathcal{F}(\rho),t''\its\phi$ and $\rho, t'\itw \psi \Leftrightarrow \mathcal{F}(\rho),t'\its\psi$ we get, for any such $t'$ and $t''$ above: 

$$\mathcal{F}(\rho),t'\its\psi \text{ and } \mathcal{F}(\rho),t''\its\phi,t''$$

and finally by the definition of $\its$ we get $$\rho,t\itw \until{\phi}{I}{\psi}\Leftrightarrow\mathcal{F}(\rho),t\its \until{\phi}{I}{\psi}$$
\end{itemize}

\end{proof}

\propinc*

\begin{proof} The proof is straightforward, we give it only for the case $\lf{\mtl}{\pww} \# \lf{\mtl}{\itww}$. \\
($\Rightarrow$) Suppose $\lf{\mtl}{\pww} \# \lf{\mtl}{\itww}$. We then have by definition $A\nsubseteq B$ and $B\nsubseteq A$ where $A = \{\lf{\phi}{\pww}\, \, |\, \,  \phi\in \mtl\}$ and $B = \{\lf{\phi}{\itww}\, \, |\, \,  \phi\in \mtl\}$. %
Neither $A$ nor $B$ is empty (since there is at least one formula in MTL), this implies that they differ at at least two elements (which, we recall, are sets), and therefore exists $\phi_1,\phi_2$ in MTL such that the element $\lf{\phi_1}{\pww}$ of $A$ does not belong to $B$ and $\lf{\phi_2}{\itww}$ of $B$ does not belong to $A$, which is the same as $\exists\phi_1,\phi_2\in\mtl\, \, \forall\psi\in\mtl$: $\lf{\phi_1}{\pww} \neq \lf{\psi}{\itww}$ and  $\lf{\phi_2}{\itww} \neq \lf{\psi}{\pww}$. \\%
($\Leftarrow$) Analogous. 
\end{proof} 

%% file: mixed.tex
\section{Proofs about mixed semantics}\label{app:mx}

\compileits*

\begin{proof} From~\cref{aux1}, and~\cref{aux2} with $t=0$ and $j=0$. %
\end{proof}

Let $\trho$ be an arbitrary compact timed word and $\rho = (\sigma,\tau) = \mathcal{C}^{-1}(\trho)$. 
We first relate every couple of indices $k$ of $\tsig$ and $0\leq j < |\tsig_k|$ with the corresponding index $i$ of $\sigma$. %
This can be easily done using the definition of $\mathcal{C}$ through the bijective function $\mathcal{m}_{\trho}$ defined as: 

$$
\mathcal{m}_{\rho}(k,j) = \begin{cases}
			j & \text{if $k = 0$}\\
            s\tau_{k-1} + j + 1 & \text{otherwise}
		 \end{cases}
$$ 

That is, $\tsig_k(j) = \sigma_{\mathcal{m}_{\trho}(k,j)}$ (we also have $\ttau_k = \tau_{\mathcal{m}_{\trho}(k,j)}$ for all $0\leq j < |\tsig_k|$). Similarly, for all $i\in \naturals$, $i<|\rho|$, we have $\sigma_i = \tsig_{\mathcal{m}^{-1}_{\trho}(i)|_{1}}(\mathcal{m}^{-1}_{\trho}(i)|_{2})$, with  $\mathcal{m}^{-1}_{\trho}(i)|_{i'}$ denoting the $i'^{\mathit{th}}$ component of $\mathcal{m}^{-1}_{\trho}(i)$. %
(The inverse bijective function $\mathcal{m}^{-1}_{\trho}$ can be easily defined as $\mathcal{m}^{-1}_{\trho}(i) = (0,i)$ if $i\leq s\tau_0$, otherwise $\mathcal{m}^{-1}_{\trho}(i) = (k,j)$ such that $k$ is the largest index satisfying $s\tau_{k-1} < i$ and $j = i - s\tau_{k-1} - 1$). 

Second, we define the right-total partial function $\mathit{ind}_{\trho}$ to relate time instants to indices in $\trho$, \ie $\mathit{ind}_{\trho}(t)$ relates $t\in \rplus$ ($t\leq \mu(\trho)$ if $\trho$ is finite) to $i$ if $\exists i : \ttau_i = \trho(t)$ (such $i$ is then unique), and $\mathit{ind_{\trho}(t)}$ is undefined otherwise. %
Accordingly, the inverse relation $\mathit{ind}^{-1}_{\trho}$ is an injective function. %
We use the standard notation $\mathit{ind_{\trho}(t)\downarrow}$ (\resp $\mathit{ind_{\trho}(t)\uparrow}$) to denote that $\mathit{ind_{\trho}}$ is defined (\resp undefined) at $t$. %

Next, we prove the following lemma. It states that the compilation $\pcompile{-}$ forces the satisfaction of a formula at points where actions occur. %
In other words, if a triple $\trho,t,j$ satisfies a compiled formula through $\pcompile{-}$ then $t$ corresponds necessarily to a timestamp in $\rho$. %

\begin{lemma}\label{aux1} For all $\phi$ in MTL, $\trho\in\tRho$, $t\in \rplus$ ($t\leq \mu(\trho)$ if $\trho$ is finite), and $0\leq j < |\trho(t)|$: $$\trho,t,j\mx  \pcompile{\phi} \Rightarrow  \mathit{ind_{\trho}(t)\downarrow}$$
\end{lemma}

\begin{proof}
By structural induction on $\phi$. 

\begin{itemize}
\item{} Case $a$: $\trho,t,j\mx  \pcompile{a}$  is equivalent to $\trho,t,j\mx  a$ which implies $\trho,t,j\nmx  \neg \varphi_\Sigma$ which in turn implies $\mathit{ind_{\trho}(t)\downarrow}$. 
\item{} The remaining cases are straightforward since $\pcompile{\phi}$ for each of them is a conjunction containing the clause $\varphi_\Sigma$.  
\end{itemize}
\end{proof}

It is perhaps easier to grasp the utility of~\cref{aux1} through its contraposition, stating that it is impossible to satisfy the compilation of any $\phi$ in MTL at time $t$ not corresponding to a timestamp in $\rho$:  $$\mathit{ind_{\trho}(t)\uparrow} \Rightarrow \trho,t,j\nmx  \pcompile{\phi}$$
which is equally easy to prove as~\cref{aux1}. %
\Cref{aux1} emphasises an important property of $\pcompile{\phi}$ that constrains the mixed semantics to only ``meaningful'' formulae in the pointwise context: for any formula $\phi$, the set of compact timed words satisfying $\pcompile{\phi}$ starting at a timepoint not corresponding to a position in $\rho$ is empty: 

$$\forall \phi\in \mtl,\, t\in\rplus, t\leq \mu(\trho) \text{ if } |\trho| < \omega: \mathit{ind_{\trho}(t)\uparrow} \Rightarrow \alf{\pcompile{\phi}}{mx} = \emptyset$$ 

It suffices now to prove the following lemma: 

\begin{lemma}\label{aux2}
For every $\phi\in\mtl$, $\trho\in\tRho$, $t\in \rplus$ such that $t\leq \mu(\trho)$ if $\trho$ is finite and $\mathit{ind}_{\trho}(t)\downarrow$, and $0\leq j < |\trho(t)|$: 
$$\trho,t,j\mx \pcompile{\phi} \Leftrightarrow \rho,\mathcal{m}_{\trho}(\mathit{ind_{\trho}(t)},j) \pw \phi$$
\end{lemma}

\begin{proof}

Let $\trho$ be an arbitrary compact timed word and $\rho = \mathcal{C}^{-1}(\trho)$. In the following, we drop accordingly the subscripts $_{\trho}$ when referring to the functions $\mathcal{m}_{\trho}$ and $\mathit{ind}_{\trho}$. %

The proof is by structural induction on $\phi$. 

\begin{itemize}
\item{} Case $a$: By the definition of $\mx$ and $\pcompile{a}$, we have 
$$\trho,t,j\mx\pcompile{a} \Leftrightarrow a = \trho(t)(j)$$
And given the bijective function $\mathcal{m}$:
$$\trho,t,j\mx\pcompile{a} \Leftrightarrow a = \sigma_{\mathcal{m}(\mathit{ind}(t),j)}$$
And finally by the definition of \pw: 
$$\trho,t,j\mx\pcompile{a} \Leftrightarrow \rho, \mathcal{m}(\mathit{ind}(t),j)\pw a$$

\item{} Case $\neg\phi$: 
By the definition of $\mx$ and $\pcompile{\neg\phi}$, we have: 
$$\trho,t,j\mx\pcompile{\neg\phi} \Leftrightarrow \trho,t,j\mx\varphi_\Sigma \text{ and } \trho,t,j\mx \neg \pcompile{\phi}$$ 
By the induction hypothesis and the fact that $\trho,t,j\mx\varphi_\Sigma$ is equivalent to $\mathit{ind}(t)\downarrow$ which is always true by the premise of~\cref{aux2} we get: 
$$\trho,t,j\mx\pcompile{\neg\phi} \Leftrightarrow \rho, \mathcal{m}(\mathit{ind}(t),j)\pw \neg\phi$$

\item{} Case $\phi\wedge\psi$: analogous to the previous case. 
\item{} Case $\until{\phi}{I}{\psi}$: By the definition of $\mx$ and $\pcompile{\until{\phi}{I}{\psi}}$:  

$$
  \begin{array}{lll}
        \trho,t,j\mx\pcompile{\until{\phi}{I}{\psi}} & \Leftrightarrow & \trho,t,j\mx \varphi_\Sigma \wedge\, \,  \exists (t,j) < (t_1,j_1)  : t_1-t\in I \wedge\, \,  \trho,t_1,j_1\mx (\pcompile{\psi} \wedge \varphi_\Sigma)\\
      &&\wedge\, \, \forall (t,j) < (t_2,j_2) < (t_1,j_1) : \trho,t_2,j_2\mx (\pcompile{\phi} \vee \neg \varphi_\Sigma)
    \end{array}
$$

Now, since $\trho,t_1,j_1\mx \varphi_\Sigma$, we have $\mathit{ind}(t_1)\downarrow$ and accordingly: 

\begin{align}\label{eq2}
\trho,t_1,j_1\mx (\pcompile{\psi} \wedge \varphi_\Sigma) \Leftrightarrow t_1 = \ttau_{\mathit{ind}(t_1)} \wedge\, \,  \trho,t_1,j_1\mx \pcompile{\psi} 
\end{align} 

And, since $\trho,t_2,j_2 \mx  \neg \varphi_\Sigma$ whenever $\rho(t_2)(j_2) = (\vdash)$ we get

\begin{align}\label{eq3}
\forall (t,j) < (t_2,j_2) < (t_1,j_1) : \mathit{ind}(t_2)\downarrow \Rightarrow \trho,\ttau_{\mathit{ind}(t_2)},j_2 \mx \pcompile{\phi}
\end{align}

Finally, given~\cref{eq2},~\cref{eq3} the induction hypotheses we get 

$$
  \begin{array}{lll}
        \trho,t,j\mx\pcompile{\until{\phi}{I}{\psi}} & \Leftrightarrow & \exists k = \mathcal{m}(\mathit{ind(t)},j) < k_1 = (\mathit{ind(t_1)},j_1)  : \tau_{k_1} -\tau_k\in I \wedge\, \,  \rho, k_1\pw \psi\\
      &&\wedge\, \, \forall k < k_2 < k_1 : \rho, k_2\pw \phi
    \end{array}
$$

\end{itemize}
\end{proof}

\supmxpw*

\begin{proof}
We first characterise the language $\ilf{\gamma_2}{\itww}$. To simplify the notation, we use $\Rho^{\omega}$ to refer to the set of all infinite timed words, \ie $\Rho^{\omega} = \{\rho\in\Rho\, \, |\, \, |\rho| = \omega\}$ and similarly $\tilde{\Rho}^{\omega}$ for the set of all infinite compact timed words. 
 
$$\ilf{\gamma_2}{\itww} = \{\rho\in\Rho^\omega\, \, |\, \, \exists t > 0 : t \in [0,1] \wedge \rho,t\models (\eventually{[0,1)}{b} \wedge \eventually{[1,1]}{c})\}$$
And by unfolding the definitions of the inner Eventually operators:  
$$\ilf{\gamma_2}{\itww} = \{\rho\in\Rho^\omega\, \, |\, \, \exists t \in (0,1] : (\exists t_1 > t: t_1 - t \in [0,1)\, \,  \wedge\, \,  \rho,t_1\models b)\, \,  \wedge\, \,  (\exists t_2 > t: t_2 - t = 1\, \,  \wedge\, \,    \rho,t_2\models c)\}$$
And we get using standard arithmetics over dense intervals:
$$\ilf{\gamma_2}{\itww} = \{\rho\in\Rho^\omega\, \, |\, \, \exists t_1,t_2 : (t_1 \in (0,2)\, \,  \wedge\, \,  \rho,t_1\models b)\, \,  \wedge\, \,  (t_2 \in (1,2]\, \,  \wedge\, \,    \rho,t_2\models c)\, \,  \wedge\, \,  t_2 - t_1 \in (0,1)\}$$
Now, since $\rho,t_1\models b$, $\rho,t_2\models c$, and $t_2 - t_1 \in (0,1)$, both $t_1$ and $t_2$ correspond to timestamps $\tau_i$ and $\tau_j$, respectively, with $j>i$: 
$$\ilf{\gamma_2}{\itww} = \{\rho\in\Rho^\omega\, \, |\, \, \exists j>i : (\tau_i \in (0,2)\, \,  \wedge\, \,  \sigma_i = b)\, \,  \wedge\, \,  (\tau_j \in (1,2]\, \,  \wedge\, \,    \sigma_j = c)\, \,  \wedge\, \,  \tau_j - \tau_i \in (0,1)\}$$

Similarly, when characterising the language $\ilf{\gamma_3}{mx}$ we get:  
$$\ilf{\gamma_3}{\mx} = \{\tilde{\rho} \in \tilde{\Rho}^{\omega}\, \,  |\, \, \exists (j,j')>(i,i') : (\tilde{\tau}_{i} \in (0,2)\,  \wedge\, \tilde{\sigma}_{i}(i') = b)\, \,  \wedge\, \,  (\tilde{\tau}_{j} \in (1,2]\, \,  \wedge\, \,    \tilde{\sigma}_{j}(j') = c)\, \,  \wedge\, \,  \tilde{\tau}_{j} - \tilde{\tau}_{i} \in (0,1)\}$$

And finally by the properties of the bijective function $\mathcal{C}$ we may link the last two characterisations: 
$$\ilf{\gamma_2}{itw} = \{\mathcal{C}^{-1}(\trho)\, \, | \, \, \trho\in\ilf{\gamma_3}{mx}\}$$
\end{proof}

\supall*

\begin{proof}
We start with the easy case: 

$$\forall \phi\in\mtl: \ilf{\phi}{\itww} \neq \{\mathcal{C}^{-1}(\trho)\, \, |\, \, \trho\in \ilf{\gamma_5}{\mxx}\}$$

The proof is straightforward using the function $\mathcal{F}$. %
Let $\rho = (\sigma,\tau) = (c,0)(b,0)(a,0)(c,5)\ldots$ and $\rho' = (\sigma',\tau') = (c,0)(a,0)(b,0)(c,5)\ldots$ such that $\sigma_i = \sigma'_i = c$ and $\tau_i = \tau'_i$ for all $i>2$ (and both $\rho$ and $\rho'$ are assumed time divergent by definition). %
Then we have $\mathcal{F}(\rho) = \mathcal{F}(\rho')$ and by~\cref{thm1} for all $\phi\in\mtl$ we have either both $\rho$ and $\rho'$ in $\ilf{\phi}{\itww}$ or none of them is. %
Now, we have $\mathcal{C}(\rho)\in\ilf{\gamma_5}{\mxx}$ (since $\mathcal{C}(\rho)$ is in $\ilf{\zeta_1 \wedge \gamma_4}{\mxx}$)  but $\mathcal{C}(\rho')\notin\ilf{\gamma_5}{\mxx}$. %

For the remaining item   

$$\forall \phi\in\mtl: \ilf{\phi}{\pww} \neq \{\mathcal{C}^{-1}(\trho)\, \, |\, \, \trho\in \ilf{\gamma_3 \vee (\zeta_1 \wedge \gamma_4)}{\mxx}\}$$

We first prove that $\ilf{\gamma_3}{\mxx}$ and $\ilf{\zeta_1 \wedge \gamma_4}{\mxx}$ are disjoint. %
This is easy to see since for any $\trho$ that satisfies $\zeta_1 \wedge \gamma_4$, we have $\trho\models\gamma_4$ and therefore for all $(t,j)$ greater than $(0,0)$ with $0<t\leq2$, we have $\trho,t,j\models\neg\varphi_\Sigma$, and in particular no action occurs within the absolute interval $(0,2]$. %
Yet, for $\trho$ to be in $\ilf{\gamma_3}{\mxx}$, it must have at least one action $b$ occurring within $(0,2)$. %
Consequently, we have  $\ilf{\gamma_3}{\mxx}\cap\ilf{\zeta_1 \wedge \gamma_4}{\mxx}=\emptyset$. %
Using the function $\mathcal{C}^{-1}$, we get:
$$\{\mathcal{C}^{-1}(\trho)\, \, | \, \, \trho\in\ilf{\gamma_3}{\mxx}\}\cap\{\mathcal{C}^{-1}(\trho)\, \, | \, \, \trho\in\ilf{\zeta_1 \wedge \gamma_4}{\mxx}\}=\emptyset$$

The second step is to show that both $\zeta_1$ and $\gamma_4$ are expressible under $\pw$. %
This part is also rather straightforward, as each of these formulae have the same semantics over $\mx$ and $\pw$ modulo $\mathcal{C}$. 
Since for every infinite compact timed word $\tilde{\rho}$ in $\ilf{\gamma_4}{\mxx}$  there is no action within $(0,2]$, $\gamma_4$ will be vacuously satisfied under $\pw$ by $\mathcal{C}^{-1}(\tilde{\rho})$ (and similarly the other way around). %

$$\ilf{\gamma_4}{\pww} = \{\mathcal{C}^{-1}(\trho)\, \, |\, \, \trho\in\ilf{\gamma_4}{\mxx}\}$$

Similarly, since for every infinite $\tilde{\rho}$ in $\ilf{\zeta_1}{\mxx}$  there is a $b$ (at a pair strictly greater than $(0,0)$) immediately followed by an $a$, $\zeta_1$ will be trivially satisfied under $\pw$ as $\mathcal{C}^{-1}(\tilde{\rho})$ does not change the timestamps or order of actions (and similarly the other way around). 

$$\ilf{\zeta_1}{\pww} = \{\mathcal{C}^{-1}(\trho)\, \, |\, \, \trho\in\ilf{\zeta_1}{\mxx}\}$$

To summarise, we have: 

\begin{equation}\label{eqq3}
$$\{\mathcal{C}^{-1}(\trho)\, \, | \, \, \trho\in\ilf{\gamma_3}{\mxx}\}\cap\{\mathcal{C}^{-1}(\trho)\, \, | \, \, \trho\in\ilf{\zeta_1 \wedge \gamma_4}{\mxx}\}=\emptyset$$ 
\end{equation}

\begin{equation}\label{eqq4}
$$\ilf{\zeta_1}{\pww} = \{\mathcal{C}^{-1}(\trho)\, \, |\, \, \trho\in\ilf{\zeta_1}{\mxx}\}$$
\end{equation}

\begin{equation}\label{eqq5}
$$\ilf{\gamma_4}{\pww} = \{\mathcal{C}^{-1}(\trho)\, \, |\, \, \trho\in\ilf{\gamma_4}{\mxx}\}$$
\end{equation}

But we also have, by~\cref{lem:exist} and~\cref{lem:itw}: 

\begin{equation}\label{eqq6}
$$\forall\psi\in\mtl: \ilf{\psi}{\pww} \neq \{\mathcal{C}^{-1}(\trho)\, \, |\, \, \trho\in\ilf{\gamma_3}{\mxx}\}$$
\end{equation}

We may now prove the remaining item 

$$\forall \phi\in\mtl: \ilf{\phi}{\pww} \neq \{\mathcal{C}^{-1}(\trho)\, \, |\, \, \trho\in \ilf{\gamma_3 \vee (\zeta_1 \wedge \gamma_4)}{\mxx}\}$$

by contradiction. Assume that 

$$\exists \phi\in\mtl: \ilf{\phi}{\pww} = \{\mathcal{C}^{-1}(\trho)\, \, |\, \, \trho\in \ilf{\gamma_3 \vee (\zeta_1 \wedge \gamma_4)}{\mxx}\}$$

This implies, using equations \ref{eqq4} and \ref{eqq5}, the following

$$\exists \phi\in\mtl: \ilf{\phi}{\pww} = \{\mathcal{C}^{-1}(\trho)\, \, |\, \, \trho\in \ilf{\gamma_3}{\mxx}\} \cup (\ilf{\zeta_1}{\pww} \cap \ilf{\gamma_4}{\pww})$$ 

Using equation \ref{eqq3} we get 
$$\{\mathcal{C}^{-1}(\trho)\, \, |\, \, \trho\in \ilf{\gamma_3}{\mxx}\} =  \ilf{\phi}{\pww} \backslash(\ilf{\zeta_1}{\pww} \cap\ilf{\gamma_4}{\pww})$$
And finally 
$$\{\mathcal{C}^{-1}(\trho)\, \, |\, \, \trho\in \ilf{\gamma_3}{\mxx}\} =  \ilf{\phi\wedge\neg(\zeta_1\wedge\gamma_4)}{\pww}$$ 
Implying that there exists a formula $\psi = \phi\wedge\neg(\zeta_1\wedge\gamma_4)$ such that 
$$\{\mathcal{C}^{-1}(\trho)\, \, |\, \, \trho\in \ilf{\gamma_3}{\mxx}\} =  \ilf{\psi}{\pww}$$
which contradicts formula~\ref{eqq6}.

\end{proof}

\compileitw*

\begin{proof}
It suffices to prove the following statement: 

For all $\phi$ in MTL, $\trho\in\tRho$ and $t\in \rplus$ (with $t\leq \mu(\trho)$ if $\trho$ is finite): 
$$\trho,t,0\mx\icompile{\phi} \Leftrightarrow \mathcal{C}^{-1}(\trho),t\itw \phi$$

We let in the following $\trho$ be an arbitrary element of $\tRho$ and $\rho = \mathcal{C}^{-1}(\trho)$. The proof is by structural induction on $\phi$. 

\begin{itemize}
\item{} Case $a$: By the definition of $\mx$ and $\icompile{a}$ we have: $\trho,t,0\mx a\vee \eventually{[0,0]}{a}$ is equivalent to $$\exists j : \trho,t,j \mx a$$
Which is equivalent, by the definition of $\mx$, to: 
$$\exists j : \trho(t)(j) = a$$
Which is equivalent, by the definition of $\trho(t)$, to:
$$\exists i,j: \ttau_i = t \wedge \tsig_i(j) = a$$
By the definition of $\mathcal{C}$, this is equivalent to:
$$\exists i : \tau_i = t \wedge \sigma_i = a$$
Which is equivalent to the following, using the definition of $\itw$:
$$\rho,t\itw a$$

\item{} Cases $\neg\phi$ and $\phi\wedge\psi$ are straightforward. 
\item{} Case $\until{\phi}{I}{\psi}$: By the definition of $\mx$ and $\icompile{\until{\phi}{I}{\psi}}$ we have $$\trho,t,0\mx \until{(\beta\Rightarrow \icompile{\phi})}{I}{(\beta \wedge \icompile{\psi})}$$
is equivalent to (by the definition of $\mx$):  
$$
\begin{array}{ll}
\exists (t,0)<(t_1,j_1): t_1-t\in I & \wedge\, \,  \trho,t_1,j_1\mx (\beta\wedge \icompile{\psi}) \\
&\wedge\, \,  \forall  (t,0)<(t_2,j_2)<(t_1,j_1): \trho,t_2,j_2\mx (\beta\Rightarrow \icompile{\phi})
\end{array}
$$
Now, $\trho,t_1,j_1\mx \beta$ is equivalent to $j_1 = 0$, and  $\trho,t_2,j_2\mx (\beta\Rightarrow \icompile{\phi})$ is equivalent to $\trho,t_2,j_2\mx \top$ for every $j_2\neq 0$. It follows that the above statement is equivalent to: 
$$
\begin{array}{ll}
\exists (t,0)<(t_1,0): t_1-t\in I & \wedge\, \,  \trho,t_1,0\mx \icompile{\psi} \\
&\wedge\, \,  \forall  (t,0)<(t_2,0)<(t_1,0): \trho,t_2,0\mx \icompile{\phi}
\end{array}
$$

And since the lexicographical order $(t,0)<(t_1,0)$ is equivalent to $t<t_1$ (and similarly for $(t,0)<(t_2,0)$ and $(t_2,0)<(t_1,0)$) we get: 
$$
\begin{array}{ll}
\exists t<t_1: t_1-t\in I & \wedge\, \,  \trho,t_1,0\mx \icompile{\psi} \\
&\wedge\, \,  \forall  t<t_2<t_1: \trho,t_2,0\mx \icompile{\phi}
\end{array}
$$

And by the induction hypotheses: 
$$
\begin{array}{ll}
\exists t<t_1: t_1-t\in I & \wedge\, \,  \rho,t_1\itw \psi \\
&\wedge\, \,  \forall  t<t_2<t_1: \rho,t_2\itw \phi
\end{array}
$$
Which is, by the definition of $\itw$, equivalent to: 
$$\rho,t\itw\until{\phi}{I}{\psi}$$
\end{itemize}

\end{proof}